\documentclass[12pt,letterpaper]{article}
\usepackage[utf8]{inputenc}
\usepackage{amsmath}
\usepackage{amsthm}
\usepackage{amssymb}
\usepackage{hyperref}
\usepackage{tabularx}
\usepackage{graphicx}
\usepackage[english]{babel}
\usepackage{color}
\usepackage{graphicx} %economic for inserting images
\usepackage{xcolor}
\definecolor{MatlabBlue}{HTML}{0072BD}   % [0.0000 0.4470 0.7410]
\definecolor{MatlabOrange}{HTML}{D95319} % [0.8500 0.3250 0.0980]
\definecolor{MatlabYellow}{HTML}{EDB120} % [0.9290 0.6940 0.1250]
\definecolor{MatlabPurple}{HTML}{7E2F8E} % [0.4940 0.1840 0.5560]
\definecolor{MatlabGreen}{HTML}{77AC30}  % [0.4660 0.6740 0.1880]
\definecolor{MatlabCyan}{HTML}{4DBEEE}   % [0.3010 0.7450 0.9330]
\definecolor{MatlabRed}{HTML}{A2142F}    % [0.6350 0.0780 0.1840]
\usepackage{amstext}
\usepackage[authoryear]{natbib}
\usepackage{verbatim}
\usepackage{listings}
\usepackage{setspace}
\usepackage{geometry}
\usepackage{caption}
\usepackage{subcaption}    % For subfigures
\usepackage{booktabs}
\usepackage{bbm}
\usepackage{amssymb}
\usepackage{float}
\usepackage[T1]{fontenc}

\newtheorem{assumption}{Assumption}
\newtheorem{proposition}{Proposition}

\newcommand{\CAR}{\mathrm{CAR}}
\newcommand{\CMR}{\mathrm{CMR}}
\newcommand{\E}{\mathbb{E}}
\newcommand{\Var}{\mathrm{Var}}
\newcommand{\Cov}{\mathrm{Cov}}
\newcommand{\Ck}{\mathcal{C}_k}
\newcommand{\indep}{\perp\!\!\!\perp}

\begin{document}

\title{Clustered Local Projections for Time-Varying Models\thanks{%
We thank seminar participants at the 2024 SNDE conference, the 2024 Midwest Econometrics Group Conference  and the Fall 2025 Midwest Macroeconomics Meeting for helpful comments. This paper is based on
research supported by the NSF under Grants No. SES-2417534 and SES-2417535. Refine.ink was used to check the paper for consistency and clarity. The authors have no conflict of interest to declare.
}}

\author{
Ana Mar\'{i}a Herrera\thanks{University of Kentucky, Email: amherrera@uky.edu}, Elena Pesavento\thanks{Emory University,
Email: epesave@emory.edu. - Corresponding author.} , Alessia Scudiero\thanks{Emory University, Email: alessia.scudiero@emory.edu},
}

\date{\today}
\maketitle
\thispagestyle{empty}
\begin{abstract}
We propose a \textit{clustered} local projection (\textit{clustered LP}) method to estimate impulse response functions in a class of time-varying models where parameter variation is linked to a low-dimensional matrix of observables. We show that the clustered LP recovers the conditional average response when the driving variables are exogenous and a weighted average of the conditional marginal effects when they are endogenous. We propose an iterative estimation method that first classifies the data using k-means, estimates impulse response functions via GMM, and evaluates differences across clustered LP estimates. Our Monte Carlo simulations illustrate the ability of \textit{clustered LP} to approximate the conditional average response function. We employ our technique to examine how uncertainty influences the transmission of a contractionary monetary policy shock to the 5- and 10-year U.S. nominal Treasury yields. Our estimation results suggest macroeconomic and monetary policy uncertainty operate through complementary but distinct channels: the former primarily amplifies the risk compensation embedded in the term premium, while the latter governs the speed and persistence with which markets revise their expectations about the future rate path following a monetary policy shock. \\

\end{abstract}

\doublespacing
\newpage \setcounter{page}{1}\parskip0.5em 

\maketitle

\section{Introduction}

Models with time-varying parameters are commonly used by applied macroeconomists and policymakers to study how shocks affect economic aggregates. Their appeal stems from the realization that many macroeconomic time series employed in policy analysis exhibit some form of nonlinearity. For instance, some researchers posit that inflation disagreement tends to increase during periods of high inflation, which, in turn, renders monetary policy less effective (see, e.g., \citet{dong2024inflation}). Others suggest that changes in the size or direction of shocks hitting the economy may lead to variation in the transmission of such shocks to aggregate economic activity (see, e.g., \citet{barnichon2022understanding}). 

A key aspect that differentiates alternative time-varying models is the law of motion that governs how the parameters evolve over time. In state-dependent models, some regime-switching models, and threshold autoregressive models, the link between parameters and external economic variables is stated explicitly, but the parameters are restricted to take values in a finite set of regimes. Other approaches, such as traditional time-varying vector autoregression models (TVP-VARs) and time-varying local projections (TVP-LP), allow the parameters to vary freely over time according to a parametric law of motion. However, they leave the connection to economic conditions unspecified. In practice, many empirical studies assume a specific law of motion for the parameters in the estimation step and, later, interpret the estimated parameter path in connection with the evolution of specific economic conditions. For example, \citet{Inoue2024} suggests that their estimated government spending multipliers are correlated with the level of public debt, while alternative regimes in Markov-Switching (MS) models are often interpreted in terms of economic expansions and recessions.

In this paper, we propose a new estimator that can be used to study impulse response functions in a broad class of time-varying parameter models, \textit{clustered local projections (clustered LP)}. In our framework, time variation is linked to a low-dimensional matrix of observables that the researcher deems relevant in explaining the evolution of the parameters via a (possibly) nonlinear function. A novelty of our framework is the use of an iterative estimator, where the data are first classified into clusters (periods) using \textit{kmeans}, the impulse responses for each cluster are estimated via local projections --thus allowing for heterogeneity over time--, and an evaluation step assesses whether the responses are statistically different between clusters.  A crucial ingredient of the clustered LP estimator is the use of a machine learning method to group together time periods when observables (e.g., economic policy uncertainty, unemployment) are similar. This contrasts with the typical literature on state-dependent models, where it is common to determine a-priori the number of states and the cutoff by grouping the data into samples below and above the mean of one variable of interest (e.g., average unemployment in \citet{Ramey2018}). 

Our paper contributes to the literature on local projections in several aspects. From a theoretical perspective, we build on the work of \citet{Goncalves2024a} and extend the idea of the state-dependent local projections to a version with many low-dimensional states. We show that the clustered LP recovers the conditional average response ($CAR$) when the variables that drive the time-variation are exogenous and a weighted average of the conditional marginal effects ($CMR$) when they are endogenous. In the latter case, impulse response estimands from the clustered LP have a causal interpretation  as in \citet{kolesar2025dynamiccausaleffectsnonlinear}. That is, the clustered LP estimands provide a scalar causal summary of the nonlinear effect in each cluster. 

From a methodological point of view, the proposed clustered LP combines a popular machine learning algorithm (k-means) with the local projection method proposed by \citet{Jorda2005}; hence, the number and partition of the states are not defined prior to the estimation. Instead, the data-driven iterative procedure consists of three steps. First, a classification step where the k-means algorithm is used to classify the data into groups (clusters) according to the similarity in the driving variables. Second, an estimation step where the local projections for all clusters and horizons are estimated jointly via a system of Generalized Method of Moments (GMM). A third step evaluates whether the impulse response functions are different across all pairs of clusters using a Wald test. We iterate on these three steps until we reject the null of equality for all clusters for a given horizon of interest.

Using \textit{clustered LP}, we inquire how effective monetary policy shocks are in altering the medium- and long-term U.S. nominal yields during uncertain times. We build on the work of \citet{DePooter2021}, \citet{Tillman2020} and \citet{Bauer2021} who find that monetary policy is less effective during periods of high monetary policy uncertainty (MPU) and \citep{aastveit2017economic} who focus on periods of high/low macroeconomic uncertainty. We deviate from their work by allowing time variation to depend on the evolution of both uncertainty measures and by letting our data-driven method determine the grouping. The iterative procedure classifies the data into four groups (low macro uncertainty-low MPU, low macro uncertainty-moderate MPU, moderate macro uncertainty-high MPU and high macro uncertainty-low MPU), which reflects the fact that the two uncertainty measures do not move in synchrony. The results suggest that on impact and in the very short-run, macroeconomic uncertainty is the main driver of heterogeneity in the response of the 5- and 10-year yields to monetary policy shocks. Nevertheless, on longer horizons, the MPU drives yield responsiveness. In fact, twelve months after the monetary policy shock, the response of the yields is an order of magnitude lower during times of high MPU. Focusing on uncertain times, we find evidence suggesting that the expectations component plays a key role in accounting for heterogeneity during times of high MPU, whereas the response of the term premium is a key determinant during periods of high macro uncertainty. All in all, our results illustrate how clustered LP responses provide a more multifaceted characterization of the interaction between monetary policy and uncertainty than state-dependent models.

\textbf{Related Literature.} \textit{Clustered LP} are most closely related to two strands of literature that estimate impulse responses in unstable environments via local projections. The first strand links the existence of parameter instability to an observable variable such as the state of the economy \citep{Ramey2018}, the direction of fiscal intervention \citet{barnichon2022understanding}, or the degree of uncertainty or disagreement about inflation expectations (\citet{Klepacz2021} and \citet{Falck2021}). An advantage of using these state-dependent models is that the researcher can explicitly investigate whether an observable macroeconomic variable of interest drives the evolution of the parameters. However, a disadvantage is that they restrict the form of time variation. 

The second strand assumes that the evolution of the parameters is independent of the observables and occurs in a deterministic fashion (\citet{Inoue2024} and \citet{MaungYoshida2025}) and the estimation is performed through local projections. Assuming a particular law of motion has several advantages: it provides a very flexible way to capture different forms of nonlinearity, it simplifies the estimation process by limiting the number of parameters to be estimated in a highly nonlinear model, and, by design, it ensures that the parameters vary gradually over time. Nevertheless, such an assumption has a limitation: the cost of remaining agnostic about the underlying source of the time variation is that the researcher may not directly learn about the economic reasons behind the changes in parameters.

Our framework is also related to a broad and expanding TVP–VAR literature that builds on \citet{Primiceri2005} and \citet{cogley2005drifts}. However, our framework differs in three important respects. First, instead of modeling parameter evolution as independent random walks, we allow the parameters to vary as nonparametric functions of observables ($Z_t$), thus explicitly linking parameter dynamics to economic conditions. Second, while TVP–VAR models assume a specific stochastic law of motion, we do not make assumptions about the functional form that links $Z_t$ with the evolution of model parameters. Third, our frequentist approach relies on LP for estimation, whereas TVP-VAR models are estimated using Bayesian methods. Similarly, our framework shares some features with the work of \citet{fischer2023}, \citet{chan2020} and \citet{hauzenberger2024} where the parameters' law of motion is linked to an observable predictor or a latent shifter. Yet, these papers also use Bayesian estimation methods. 

Finally, our empirical framework bears some resemblance to Markov Switching (MS) models (e.g., \citet{hamilton1989new} and \citet{diebold1994regime}) where the parameters are assumed to switch over time according to some (possibly unobserved) variable, $Z_t$. However, in contrast to MS models where the number of states is determined a priori, the classification step of our algorithm uses k-means to select the number of states and partition the data accordingly.  Our model is similar, although conceptually different, to a functional-coefficient autoregressive model (see, e.g. \citet{chen1993functional}). Unlike the standard functional-coefficient we allow the parameters to be driven by an external variable rather than lagged dependent variable. More importantly, we focus on local projections regressions as we are interested in estimating impulse responses.

\textbf{Outline.} The remainder of the paper proceeds as follows. Section \ref{sec:Estimation} outlines the \textit{clustered LP} framework and describes the estimator. Section \ref{sec:simulations} illustrates the performance of the estimator via Monte Carlo simulations for alternative DGPs; Section \ref{sec:application} uses \textit{clustered LP} to inquire about the effect of monetary policy shocks on the 5- and 10-year yields during uncertain times. Section \ref{sec:conclusions} concludes with a brief summary of the paper and possible future steps.

%XXXXXXXXXXXXXXXXXXXXXXXXXXXXXXXXXXXXXXX 
\section{Clustered local projections}\label{sec:Estimation}

In this section, we propose a fast and easy-to-implement method, \textit{clustered LP}, that leverages the ease of implementation of LP and k-means to estimate time-varying impulse response functions when time variation is driven by an observable variable of interest. Our framework relies on a first step where heterogeneity over time is revealed, allowing us to classify the data into groups, and then the responses to impulse response functions are estimated via LP. 

\subsection{Theoretical Framework}
Our baseline specification is a time-varying local projection where the evolution of the parameters is driven by a variable (or low-dimensional set of variables), $Z_{t-1}$, given by:

\begin{equation}\label{eq:lp_np}
y_{t+h}=\beta_t^h(Z_{t-1}) \varepsilon_t+\gamma_{t}^{h}(Z_{t-1})^{\prime} \mathbf{w_t}+v_{t+h}
\end{equation}

\noindent where $t=1, \ldots, T$; $y_{t+h}$ is a scalar endogenous variable of interest $h$ periods ahead, $\varepsilon_t$ is the structural shock of interest, $\beta_t^h$ denotes the time-varying impulse response at horizon $h$, $\mathbf{w_t}$ is a vector of control variables including deterministic terms and lags of endogenous variables, and $v_{t+h}$ is the local projection residual. The notation $ \beta_t^h(Z_{t-1})$, and $\gamma_{t}^{h}(Z_{t-1})$ makes explicit that the model parameters vary as unknown functions of the external variables $Z_t$ at time $t-1$. The timing of the variables $Z_{t-1}$ follows the literature on state dependent models by assuming that the value of $Z$ in the previous time period affects the current value of the parameters. 

We make the following two assumptions:\footnote{The use of more primitive conditions is outside the scope of this paper; for example, \citet{chen1993functional} gives conditions for geometric ergodicity for a univariate version of our model when $Z_t=Y_{t-p}$.}  

\begin{assumption}[Identification]\label{as:1}
$\varepsilon_{t}$ is an observed structural shock such that for each $t$:
\begin{enumerate}
\item[(a)] $\varepsilon_t$ is independent of $\mathcal{F}_{t-1} = \sigma(y_{t-j}, \mathbf{w_t}, Z_{t-j}, \varepsilon_{t-j} : j \geq 1)$ and of $(\varepsilon_{t+1}, \ldots, \varepsilon_{t+H})$;
\item[(b)] $\varepsilon_t$ is continuously distributed on an interval $I \subseteq \mathbb{R}$ with mean zero and positive, finite variance.
\end{enumerate}
\end{assumption}

\begin{assumption}[Estimation]\label{as:2}
$y_t$  and $Z_t$ are strictly stationary and ergodic. 
\end{assumption}

Assumption~\ref{as:1} states the conditions required for the identification result in Proposition~\ref{prop:CLP}. It requires full independence of $\varepsilon_t$ from $\mathcal{F}_{t-1}$, not merely mean independence; this rules out GARCH-type conditional heteroskedasticity but does not restrict the unconditional variance to be constant. Assumption~\ref{as:2} is not required for identification but ensures that the sample OLS estimator converges to the population quantity $\beta_k^h$. Under Assumptions~\ref{as:1} and~\ref{as:2}, the distribution of $\varepsilon_t$ is time-invariant (by stationarity of the system), so $\varepsilon_t$ is in fact i.i.d.\ and constant variance holds as a consequence rather than as a separate restriction.

For ease of exposition, in the remainder of the paper, we assume that the shock $\varepsilon_t$ has been appropriately identified and we focus on the response of the variable $y_t$ at time $t+h$ to a shock of size $\delta$ in $\varepsilon_{t}$. Nevertheless, the \textit{clustered LP} approach we propose is flexible and can easily accommodate alternative identification strategies, such as instrumental variables or externally identified shocks, as in our empirical application in Section \ref{sec:application}.  

In our theoretical framework, we consider two DGPs. In the first,  we focus on  the case where $Z_{t-1}$ is a low-dimensional set of exogenous driving variables.  In the second, we allow for endogenous $Z_{t-1}$. In both cases, we first classify the  observations into \(K\) `groups/clusters' where the driving variable, $Z_{t-1}$, takes on similar values. Note that the timing of $Z_{t-1}$, which is crucial for our model, implies that it is predetermined with respect to the shock at time \(t\). 

Suppose that the data has been classified into $K$ clusters using k-means clustering. Then, the \textit{clustered LP} for the variable $i$ at horizon $h$ is given by
\begin{equation}\label{eq:CLP}
  y_{t+h} = \sum_{k=1}^{K} D_k\, \beta_{k}^h\,\varepsilon_{t}+  \sum_{k=1}^{K} D_k\, \gamma_{k}^{h}\, \mathbf{w_t}
              + v_{t+h},
\end{equation}
where $D_k = \mathbf{1}\{Z_{t-1}\in\Ck\}$ is the indicator for the cluster~$k$ (with $\Ck$ the $k$-th k-means cluster on the support of $Z_{t-1}$), $\mathbf{w_t}$ collects all controls, and the $i$ subindex for variable $i$ is omitted for simplicity. 

Following \citet{Goncalves2024a},  we start by defining the causal objects of interest in our time-varying setting. Let $y_{t+h}(e)$ denote the potential outcome obtained by setting $\varepsilon_{t}=e$. More precisely, the outcome admits the nonparametric structural representation
\begin{equation}\label{eq:structural}
  y_{t+h} \;=\; \psi_h(\varepsilon_t,\, U_{h,t+h}),
\end{equation}
where $\psi_h(\cdot,\cdot)$ is an unknown measurable function and $U_{h,t+h}$ collects all variables other than $\varepsilon_t$ that
causally affect $y_{t+h}$: initial conditions ($y_{t-1}$ and lags), the
path of the driving variable ($Z_{t-1}, Z_t, \ldots, Z_{t+h-1}$), and all
other disturbances entering the outcome between periods $t$ and $t+h$. The \emph{Conditional Average Structural Function} for cluster~$k$ is given by
\begin{equation}
  \Psi_k^h(e) \;\equiv\;
  \E\!\left[y_{t+h}\mid \varepsilon_{t}=e,\; Z_{t-1}\in\Ck\right],
\end{equation}

We define two causal objects. The \emph{Conditional Average Response} for cluster~$k$ to a shock of
fixed size~$\delta$ is given by
\begin{equation}
  \CAR^h(\delta,k) \;\equiv\;
  \E\!\left[y_{t+h}(\varepsilon_{t}+\delta)-y_{t+h}(\varepsilon_{t})
        \,\Big|\, Z_{t-1}\in\Ck\right],
\end{equation}

and the \emph{Conditional Marginal Response} is given by $CMR^h(k)=\frac{C A R^h(\delta,k)}{\delta}=CAR^h(1,k).$

% \begin{equation}\label{eq:CMR}
% CMR^h(k)=\frac{C A R^h(\delta,k)}{\delta}=CAR^h(1,k).
% \end{equation}

In the context of a state dependent model, \citet{Goncalves2024a} shows that the state-dependent LP estimates the $CAR$ when the state is exogenous. Thus, when $Z_{t-1}$ is exogenous, the \textit{clustered LPs} recover the $CAR$. More precisely, for $h=0$ the \textit{clustered LP} estimates the impact effect of the shock $\varepsilon_{1t}$ conditional on the value of $Z_{t-1}$ for that particular cluster. For longer horizons, $\hat{\beta}_{k}^h$ estimates the average response over all possible future paths of the state between time $t$ and $t+h$. When  $Z_{t-1}$ is endogenous, the estimates still have a casual interpretation. In particular, in the state-dependent case, \citet{Goncalves2024a} shows that under normal errors, linear LPs recover the $CMR$. Similarly, the \textit{clustered LP} recovers the $CMR$. When normality fails, \citet{kolesar2025dynamiccausaleffectsnonlinear} demonstrate that linear LPs admit a causal interpretation even when the underlying model is nonlinear. This is also the case here.

The following proposition formally extends these results to our \textit{clustered LP}.

\begin{proposition}\label{prop:CLP}
  Consider  the clustered local projection
  \eqref{eq:CLP}. Fix any cluster $k\in\{1,\ldots,K\}$ formed by k-means
  clustering of $Z_{t-1}$.

  \medskip
  \noindent\textbf{Part~(i) - General DGP.}
  Suppose Assumption~\ref{as:1} holds and the following
  regularity conditions are satisfied for each cluster~$k$:
  \begin{enumerate}
    \item[(a)] $g_{h,k}(e)\equiv \E[y_{t+h}\mid\varepsilon_{t}=e,
      D_k=1]$ is locally absolutely continuous in $e$ on $I$;
    \item[(b)] $\E[|g_{h,k}(\varepsilon_{t})|(1+|\varepsilon_{t}|)\mid D_k=1]
      <\infty$ and $\int_I \omega_k(e)|g_{h,k}'(e)|\,de<\infty$, where
      $\omega_k$ is defined below.
  \end{enumerate}
  Then  for each cluster~$k$ and
  all horizons $h\geq 0$:
  \begin{equation}\label{eq:part1}
    \beta_k^h \;=\; \int_I \omega_k(e)\,\Psi_{k}^{h\prime}(e)\,de,
  \end{equation}
  where $\omega_k(e) \;\equiv\;
    \frac{\Cov\!\left(\mathbf{1}\{\varepsilon_{t}\geq e\},\,\varepsilon_{t}
          \,\Big|\, D_k=1\right)}
         {\Var(\varepsilon_{t}\mid D_k=1)}.$   The weight function $\omega_k$ is: (i)~nonneg\-a\-tive and integrates to
  one; (ii)~hump-shaped, peaking near $\E[\varepsilon_{t}\mid D_k=1]$; and
  (iii)~determined solely by the conditional distribution of $\varepsilon_{t}$
  given cluster~$k$, not by the outcome variable or horizon~$h$.

  \medskip
  \noindent\textbf{Part~(ii) - Exogenous $Z$.}
  Suppose, in addition to Assumption~\ref{as:1}, that $\left\{Z_s: s \geq t\right\} \perp \varepsilon_t$, then for each cluster~$k$ and
  all horizons $h\geq 0$:
  \begin{equation}\label{eq:part2}
    \beta_k^h \;=\; \frac{\CAR^h(\delta,k)}{\delta}
              \;=\; \CAR^h(1,k),
  \end{equation}
  which equals $\CMR^h(k)$, the conditional marginal response for
  cluster~$k$.
\end{proposition}

 Proposition~\ref{prop:CLP} is a population-level result that holds for any partition of $Z_{t-1}$ into $K$ clusters. The data-generating process need not have discrete regimes: when $\beta^h(Z_{t-1})$ varies continuously, the clustered LP provides a piecewise-constant approximation and Proposition~\ref{prop:CLP} ensures that each piece retains a causal interpretation. The proof of Proposition~\ref{prop:CLP} can be found in Appendix \ref{sec:appendix}; the formal approximation properties of this estimator as $K$ grows with $T$ are left for future work. 

Some remarks are in order here. First, we follow the usual timing convention of state dependent models: $D_k = \mathbf{1}\{Z_{t-1}\in\Ck\}$ is a function of $Z_{t-1}$, which is predetermined with respect to $\varepsilon_{t}$. Under Assumption~\ref{as:1} (a), $\varepsilon_{t}$ is independent of $\mathcal{F}_{t-1}$ and hence independent of the past, so
  $\varepsilon_{t}\indep D_k$ for all $k$.  This makes the within-clustered LP algebraically equivalent to a standard LP estimated on the subsample (i.e., cluster) defined by $Z_{t-1}$. Second, we do not make any assumptions regarding the functional form of the true LP within the cluster. However, we note that when $Z_{t-1}$ is exogenous, the evolution of $Z$ is independent of the structural shock $\varepsilon_t$, so that future state variables do not respond to $\varepsilon_t$.  In this case the potential outcome is linear in $e$ and a linear regression within each cluster recovers $CAR$. When $Z_{t-1}$ is endogenous, Part (i) applies Proposition 1 of
  \citet{kolesar2025dynamiccausaleffectsnonlinear} to the within-cluster subsample so the estimator admits a causal interpretation. Given the timing of the model, the clustered LP will recover the $CAR$ only at impact. For $h>0$, the OLS estimand can be expressed as a weighted average of marginal effects $\Psi_{k}^{h \prime}(e)$ with weights $\omega_k(e)$ that can be estimated from the data. 
  When the weights are positive, the estimand inherits the sign of the $CMR$ within the cluster whenever $\Psi_{k}^{h \prime}(e)$ has a constant sign.\footnote{See the discussion of \citet{kolesar2025dynamiccausaleffectsnonlinear} and the replies to the discussants} Lastly, a researcher interested in recovering $CAR$ could employ the nonparametric approach of \citet{Goncalves2024b}.

\subsection{Estimator}\label{sec:GMM}
 The \textit{clustered LP} estimator consists of three iterative steps: classification, estimation, and evaluation. Note that the estimation of $\beta_k^h$ requires, in addition to the identification conditions in Assumption~\ref{as:1}, that the data are strictly stationary and ergodic (Assumption~\ref{as:2}) so that the sample OLS/GMM estimator converges to the population quantity characterized in Proposition~\ref{prop:CLP}.

\textbf{Classification step:} We first partition the whole sample into groups using k-means. More specifically, we start with a maximum number of clusters, $K$, and use the k-means algorithm to classify $Z_{t-1}$ into groups (clusters) based on their similarity.\footnote{For more details on Lloyd's algorithm, see Lloyd, Stuart P. “Least Squares Quantization in PCM.” IEEE Transactions on Information Theory, vol. 28, 1982, pp. 129–137.}

\textbf{Estimation step:} Then, building on \citet{inouejordakuersteiner2024}, we estimate the local projection for all clusters and horizons jointly via  a system Generalized  Method of Moments (GMM). Specifically, we rewrite equation (\ref{eq:CLP}) in matrix form as a system of $H+1$ equations:
\begin{equation}\label{eq:GMM}
\begin{aligned}
\boldsymbol{y}_t(H)
&=\mathbf{C(H)}\mathbf{x}_t+ \boldsymbol{v}_t(H)
\end{aligned}
\end{equation} 
where $\mathbf{D}_t':=\left[ D_{1t} \ldots D_{Kt} \right]$, $\mathbf{w_t}$ is an $m \times 1$ vector that contains control variables,
\begin{align}
\mathbf{C}(H) = \begin{bmatrix}
\mathbf{B}(H) & \mathbf{\Gamma}(H)
\end{bmatrix},
\quad
\mathbf{x}_t = \begin{bmatrix}
\mathbf{D}_t \otimes \varepsilon_{t} \\
\mathbf{D}_t \otimes \mathbf{w_t}
\end{bmatrix},
\end{align}
$\boldsymbol{y}_t(H)$, and $\boldsymbol{v}_t(H)$ are vectors with variables stacked over the horizons $h=0,\ldots,H$, such that $\boldsymbol{y}_t(H) = (y_t, y_{t+1}, \dots, y_{t+H})^\top \in \mathbb{R}^{(H+1) \times 1}$
and $\boldsymbol{v}_t(H) = (v_t, v_{t+1}, \dots, v_{t+H})^\top \in \mathbb{R}^{(H+1) \times 1}$. $\mathbf{B(H)}$ and $\mathbf{\Gamma(H)}$ are matrices that collect the parameters stacked accordingly, and the first $K \times (H+1)$ parameters, $\beta_k^h$, correspond to the impulse response estimates for each of the $K$ clusters. The population moment conditions of the system of clustered LP can be written as
$$\mathbb{E}\left[\mathbf{x}_t \cdot\left(\boldsymbol{y}_t(H)-\mathbf{C(H)}\mathbf{x_t}\right)^{\prime}\right]=0.$$
Defining the moment function as $m_t(\theta)=\mathbf{x}_t \otimes\left(\boldsymbol{y}_t(H)-\mathbf{C(H)}\mathbf{x_t}\right),$
allows us to express the GMM estimator as
\begin{equation}\label{eq:GMMest}
\hat{\theta}=\arg \min _\theta\left[\frac{1}{N} \sum_{t=t_0}^{T^*}  m_t(\theta)\right]^{\top} \hat{V}^{-1}\left[\frac{1}{N} \sum_{t=t_0}^{T^*}  m_t(\theta)\right]
\end{equation}
where $N=T^*-t_0$, $t_0$ denotes the first observation available after accounting for lags in the control set, $T^*=T-H-1$, and $\hat{V}$ denotes the optimal weighting matrix, correcting for heteroskedasticity and autocorrelation. Joint estimation of the IRFs using this GMM setup is useful, as it provides an estimate of the covariance matrix, $\hat{\Omega}_\beta$, a key ingredient in the construction of a joint test on the IRFs to determine the optimal number of clusters in the next step.  

\textbf{Evaluation step:} After estimating the IRFs for the $K$ clusters, we evaluate whether they differ across clusters via a series of pairwise Wald tests.\footnote{See \citet{kilianvigfusson2011} for a similar pairwise test in a linear setup.} Specifically, given a horizon of interest to the analyst, ($\tilde H$), which may differ from $H$, let the null hypothesis of equality between the responses at horizon $h=0,1,\dots,\tilde H$ for clusters $k$ and $k'$ be given by
\begin{equation*}
   H_0: \beta_k^h - \beta_{k'}^h=0 \text{ for }h=0,1,2,...,\tilde H.
\end{equation*}

Let $\beta= [\begin{array}{llllll} \beta_k^0 &\cdots & \beta_k^H & \beta_{k'}^0 &\cdots & \beta_{k'}^{\tilde{H}}]'\end{array}$, then the Wald statistic can be expressed as
\[
W \;=\; (R\hat\beta)^\prime \bigl(R\,\hat \Omega_\beta\,R^\prime\bigr)^{-1} (R\hat\beta )
\;\sim\; \chi^2_{\tilde{H}+1} \quad \text{under } H_0.
\]

To control for multiple testing across all $\tfrac{K(K-1)}{2}$ possible cluster pairs, we apply a Bonferroni correction by adjusting the significance level to $\alpha_{\text{adj}} = \alpha / [K(K-1)/2]$. Each pairwise comparison is evaluated against the corresponding critical value $\chi^2_{r,1-\alpha_{\text{adj}}}$. If we fail to reject the null for a pair, then we repeat the procedure with $K-1$ clusters. We iterate until we reject the null for all pairs.
At each iteration of the procedure, the Bonferroni correction ensures that the probability of any false rejection among the pairwise Wald tests is at most $\alpha$. Since the procedure performs at most $K - 1$ iterations, the probability that any false rejection occurs is at most $(K - 1)\alpha$. In practice, this bound is conservative as the procedure typically stops before exhausting all $K - 1$ stages.

It is important to note that the iterative procedure does not aim to recover a ``true" number of clusters. In general, the data-generating process need not have discrete regimes as the IRF $\beta_t^h(Z_{t-1})$ may vary continuously with the driving variable. In this case, the clustered LP provides a piecewise-constant approximation, and the selected $\hat{K}$ reflects the number of groups whose IRFs are statistically distinguishable at the given sample size and significance level. Two clusters with similar, but not identical, impulse responses will be merged if the difference is not detectable. Conversely, as the sample size grows and estimation precision improves, finer differences become detectable and $\hat{K}$ may increase.

\section{Illustrative Simulations}\label{sec:simulations}

This section presents simulation results that illustrate the performance of the \textit{clustered LP}. We focus on the case where $Z_t$ is exogenous, thus our objects of interest are the number of states selected by the Wald test and the $CAR^h(\delta,k)$. 

We consider a simplified bi-variate model given by:
\begin{equation}\label{eq:sim}
\begin{cases}
x_t = \varepsilon_t^x, \\
y_t = \beta_t\,(Z_{t-1})x_t + \gamma_{1,t}(Z_{t-1})\, y_{t-1} + \gamma_{2,t}(Z_{t-1})\, y_{t-2}
        + \varepsilon_t^y, 
\end{cases}
\end{equation}
where we assume the shock of interest, $\varepsilon_t^x$, is observed; the intercepts have been normalized to zero; the innovations, $(\varepsilon_t^x, \varepsilon_t^y)^\top \sim \mathcal{N}(\mathbf{0}, V)$
with $V = I_2$, and parameter evolution is driven by a low-dimensional set of exogenous variables $Z_{t-1}$. 

We present simulations where time-variation is driven by three different DGPs. We employ  $M=10000$ Monte Carlo replications and set $T=2000$. As is common in time-varying models, we are interested in the IRFs conditional on the time when the shock hits and the initial conditions given by $Z_{t-1}$. Therefore, to simulate the object of interest, we proceed as follows.

\textbf{Step 1: Baseline time series for $y^{(m)}_t$.} For each Monte Carlo replication ($m=1,2,\dots,M$), we generate a time series $\{y^{(m)}_t\}_{t=1}^{T}$ from the initial conditions, $(\beta^{(m)}_1, \gamma^{(m)}_{1,1}, \gamma^{(m)}_{2,1})$ using the DGP for $Z_{t}$ and the model (\ref{eq:sim}). We discard $10{,}000$ observations to ensure a stationary distribution.

\textbf{Step 2: Computation of the $\mathbf{CAR}$.} The object of interest is the response of $y_{t+h}$, for $h=0,1,..,H$, to a shock of size $\delta=1$ hitting $\varepsilon^x_t$ at time $t$ where, to mimic the object of interest in time-varying models, $t$ is any point in the sample. Thus,  for every $t$,  we simulate two paths for the outcome variable: a counterfactual
baseline path where $x_t=\varepsilon_t^x$ and a perturbed path where $x_t=\delta+\varepsilon_t^x$. Then, for each group $k$ we compute
\begin{equation*}
    CAR^h(\delta,k)=\frac{1}{M}\sum_{m=1}^{M}\!\left[y^{(m)}_{t+h}(\varepsilon^x_{t}+\delta)-y^{(m)}_{t+h}(\varepsilon^x_{t})
        \,\Big|\, z^{(m)}_{t-1}\in\Ck\right].
\end{equation*}
\textbf{Step 3: Clustered LP estimation.} The clustered LP estimates are obtained using the iterative algorithm described in section \ref{sec:GMM}. The initial number of groups is set to $K = 10$, the number of horizons for the Wald test is set to $\tilde H=5$, and the significance level is $\alpha=0.05$.

\subsection{Smooth Transition Threshold Models}

We consider two smooth transition threshold models.

\textbf{Univariate Case:} Let the DGP for the time-varying parameters be given by
\begin{equation*}
\boldsymbol{\psi}_t(z_{t-1}) = \sum_{k=1}^{K} \xi_k(z_{t-1})\,\boldsymbol{\psi}_k + \boldsymbol{\eta}_t
\end{equation*}
where $\boldsymbol{\psi}_t = (\beta_t,\, \gamma_{1,t},\, \gamma_{2,t})'$ 
and $\boldsymbol{\eta}_t = (\eta_{1,t},\, \eta_{2,t},\, \eta_{3,t})'$, the number of regimes is $K=4$, $\xi_k(z_{t-1})$ are weights associated with each regime, and $\eta_{i,t}$, $i=1,2,3$ are uncorrelated i.i.d. normal errors with variance set to $0.0009$.\footnote{A small variance is required for the model to be stationary.} To mimic time dependence -- often found in macro and financial data -- the threshold variable $z_t$ is drawn from a stationary
$\mathrm{ARMA}(p_z, m_z)$ process:
\[
  z_t = c_z + \sum_{i=1}^{p_z} \phi_i^z\, z_{t-i}
        + \varepsilon_t^z + \sum_{j=1}^{m_z} \theta_j^z\, \varepsilon_{t-j}^z,
  \qquad
  \varepsilon_t^z \sim \mathcal{N}(0, \sigma_e^2).
\]
The calibration uses $p_z = 2$, $m_z = 3$, $c_z = 0$, AR parameters
$(\phi_1^z, \phi_2^z) = (0.6, 0.3)$, MA parameters
$(\theta_1^z, \theta_2^z, \theta_3^z) = (0.8, 0.7, 0.4)$,
and $\sigma_e = 1$.

Four latent regimes are defined by the empirical quartiles of $z_t$,
giving three thresholds $(\tau_1, \tau_2, \tau_3)$ at the 25th, 50th,
and 75th percentiles.
To ensure stationarity, the regime-specific parameters are set to
\[\begin{bmatrix}
\beta_1 & \beta_2 &  \beta_3 & \beta_4 \\
\gamma_{1,1} & \gamma_{1,2} & \gamma_{1,3} & \gamma_{1,4} \\
\gamma_{2,1} & \gamma_{2,2} & \gamma_{2,3} & \gamma_{2,4} \\
\end{bmatrix}
=
\begin{bmatrix}
-1.9 & -0.5 & 0.2 & 0.8 \\
0.7 & 0.4 & 0.9 & 1.2 \\
0.1 & 0.2 & -0.1 & -0.3 \\
\end{bmatrix}
\]
and the change between regimes is governed by the transition logistic function
\begin{equation}
\label{eq:logistic}
G_k(z_{t-1};\lambda) \;=\; \frac{1}{1+\exp\!\big(-\lambda\,[\,z_{t-1} - c_k\,]\big)}\,,
\qquad \lambda>0,
\end{equation}
where setting $\lambda=5$ ensures a smooth transition between non absorbing regimes, $(c_1,c_2,c_3)=(-4.3359, -0.5981 ,3.5717)$. The regime weights are constructed sequentially as $\xi_1(z_{t-1}) = 1 - G_1(z_{t-1};\lambda)$,
$\xi_k(z_{t-1}) = G_{k-1}(z_{t-1};\lambda) - G_k(z_{t-1};\lambda)$, for $k=2,3$,and
$\xi_4(z_{t-1}) = G_3(z_{t-1};\lambda)$, so that $\xi_k \in [0,1]$ and
$\sum_{k=1}^{4}\xi_k(z_{t-1})=1$ for all $z_{t-1}$.

\textbf{Bivariate Case:} The second DGP we consider assumes that $Z_{t-1}$ is bivariate and generated by a stationary VARMA(2,3). As in the univariate case, we assume that there are four regimes and that the transition across regimes for each of the variables in $Z_t$ is a logistic transition function. We use of a single threshold for each driving variable to generate four states (see the Online Appendix B for details).

\begin{table}[htbp]
\captionsetup{position=top}\caption{Frequency of Clusters from Iterative Procedure}
\label{tab:k_counts}
\centering
\begin{tabular}{c c c c}
\hline
& Univariate & Bivariate & Absolute \\
$\hat{K}$ & Threshold & Threshold & Value\\
\hline
2 & 0\% & 0\% & 13.51\% \\
3 & 0\% & 12.9\%  & 73.14\%\\
4 & 91.4\% & 86.2\% & 12.8\% \\
5 & 7.3\% & 0.8\% & 0.55\%\\
6 & 1.1\% & 0.07\% & 0\%\\
7 & 0.2\% & 0.04\% & 0\%\\
\hline
\end{tabular}
\par\smallskip
\flushleft\footnotesize{\textit{Notes:} This table reports the frequency of the number of clusters, $\hat{K}$, estimated with the iterative procedure across $10{,}000$ Monte Carlo replications for the univariate and bivariate threshold DGPs, as well as for the absolute value DGP. } 
\end{table}

Table \ref{tab:k_counts} reports the frequency of $\hat{K}$ selected. The procedure selects four clusters 91.4\% (86.2\%) of the time in the univariate (bivariate) model. When the a number of clusters selected differs from the true  ($K=4$), the procedure selects a larger number of clusters ($\hat{K}=5$ for 7.3\% of the simulations) for the univariate DGP, but favors a parsimonious model for the bivariate DGP ($\hat{K}=3$ for 12.9\% of the simulations). Recall that the iterative procedure involves a classification step that is based on $Z_{t-1}$ and an evaluation step that tests differences in the IRFs up to an horizon of interest for the researcher. Therefore, the number of clusters selected can differ from the number of clusters in $Z_{t-1}$. We will return to this point in section \ref{sec:small}.

The purpose of the clustered LP is to reduce the dimensionality of the estimation problem so as to facilitate computation of the conditional average response in general time-varying setups while providing a causal summary of the time-varying effects. Figure \ref{fig:uniSTM} illustrates the IRF (gray) for each value of $z_{t-1}$ and the partition of the IRFs when the number of groups is set to the true $K=4$ or the next most likely number of groups ($K={5,6}$) selected by the iterative procedure. The figure also plots the $CAR^h(1,k)\in C_k$ (red), and the average response estimated by the clustered LP (purple) for each group. Two insights are derived from this figure. First, as expected, the iterative procedure divides the data into different subsamples based on $Z_{t-1}$ and the differences in the IRFs across groups up to $\tilde H$. Due to the stationary nature of the DGP, the largest differences are observed on impact. Second, the clustered LP accurately recover the average conditional response for each group. Note how the clustered LP estimates are indistinguishable from the $CAR^h(1,k)\in C_k$.

\begin{figure}[htbp]
    \centering
    \captionsetup{position=top}
    \caption{Univariate Threshold Model}
    \label{fig:uniSTM}
    \begin{subfigure}[t]{0.33\textwidth}
        \centering
        \includegraphics[width=\textwidth]{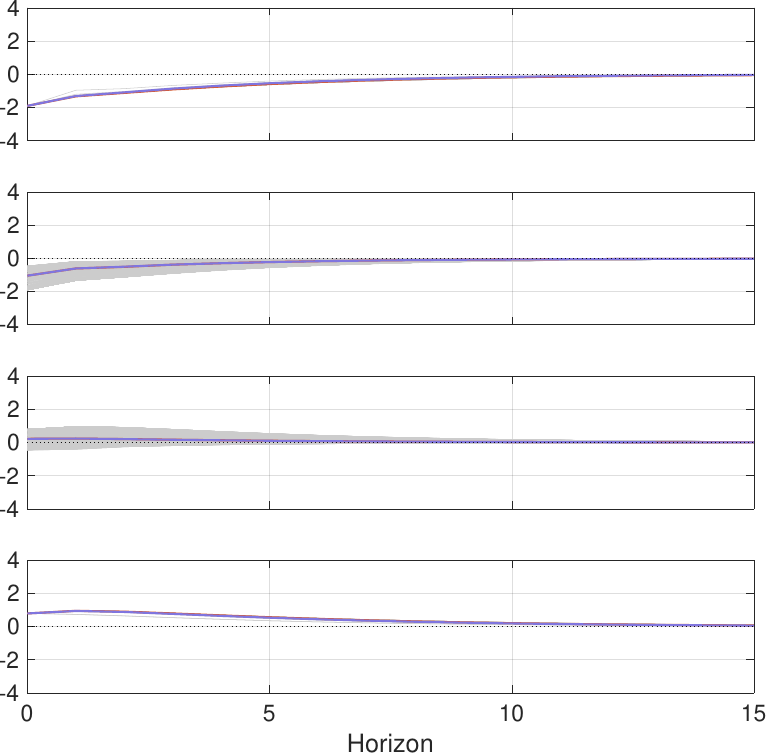}
        \label{fig:univ_smooth_3k}
    \end{subfigure}%
    \hfill%
    \begin{subfigure}[t]{0.33\textwidth}
        \centering
        \includegraphics[width=\textwidth]{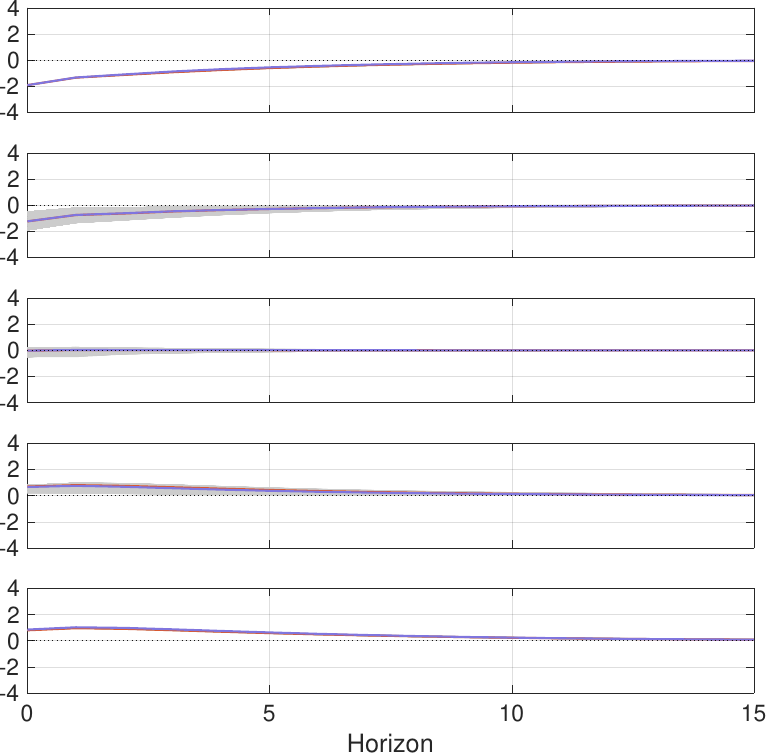}
        \label{fig:univ_smooth_4k}
    \end{subfigure}%
    \hfill%
    \begin{subfigure}[t]{0.33\textwidth}
        \centering
        \includegraphics[width=\textwidth]{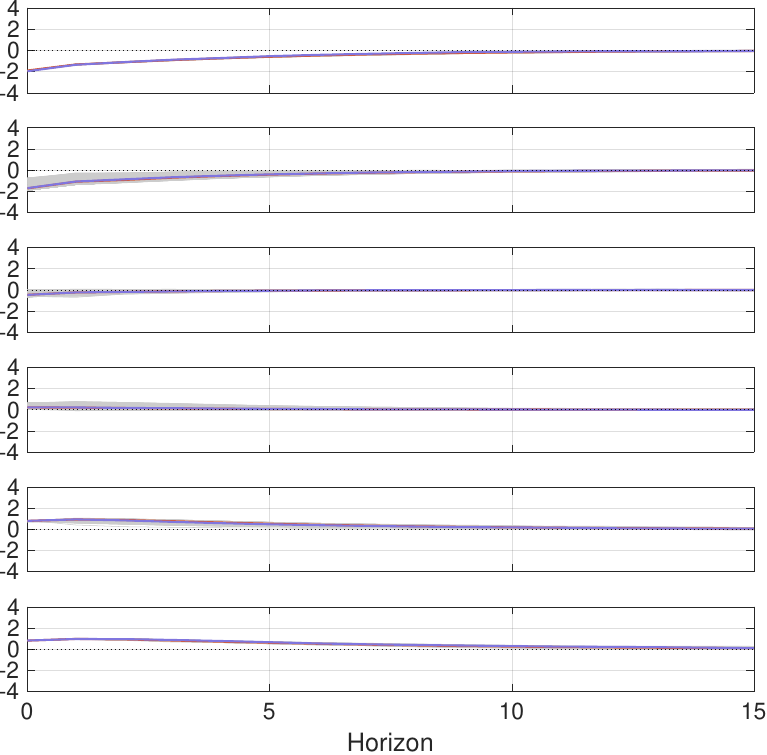}
        \label{fig:univ_smooth_5k}
    \end{subfigure}%
    \vspace{6pt}
    \begin{minipage}{\textwidth}
        \flushleft\footnotesize{\textit{Notes:} This figure illustrates the true impulse response functions for each $Z_{t-1}$ in gray, the clustered LP estimate in (purple), and the $CAR^h(\delta=1,k)$ (in red) for each partition ($K=4$ in the left panel, $K=5$ in the middle panel, and $K=6$ in the right panel).}
    \end{minipage}
\end{figure}

To illustrate a key advantage of our approach, the ability to have a low-dimensional multivariate driving force $Z_{t-1}$, Figure \ref{fig:biSTM} plots the simulation results for the bivariate smooth threshold model. The figure shows the IRF (gray) for each value of $z_{t-1}$ and the partition of the IRFs when the number of groups is set to the true $K=4$, the $CAR^h(1,k)\in C_k$ (red), and the average response estimated by the clustered LP (purple) for each group. The figure also reports simulation results for the next most likely number of groups ($K={3,5}$) selected by the iterative procedure. 

\begin{figure}[htbp]
    \centering
    \captionsetup{position=top}
    \caption{Bivariate Threshold Model}
    \label{fig:biSTM}
    \begin{subfigure}[t]{0.32\textwidth}
        \centering
        \includegraphics[width=\textwidth]{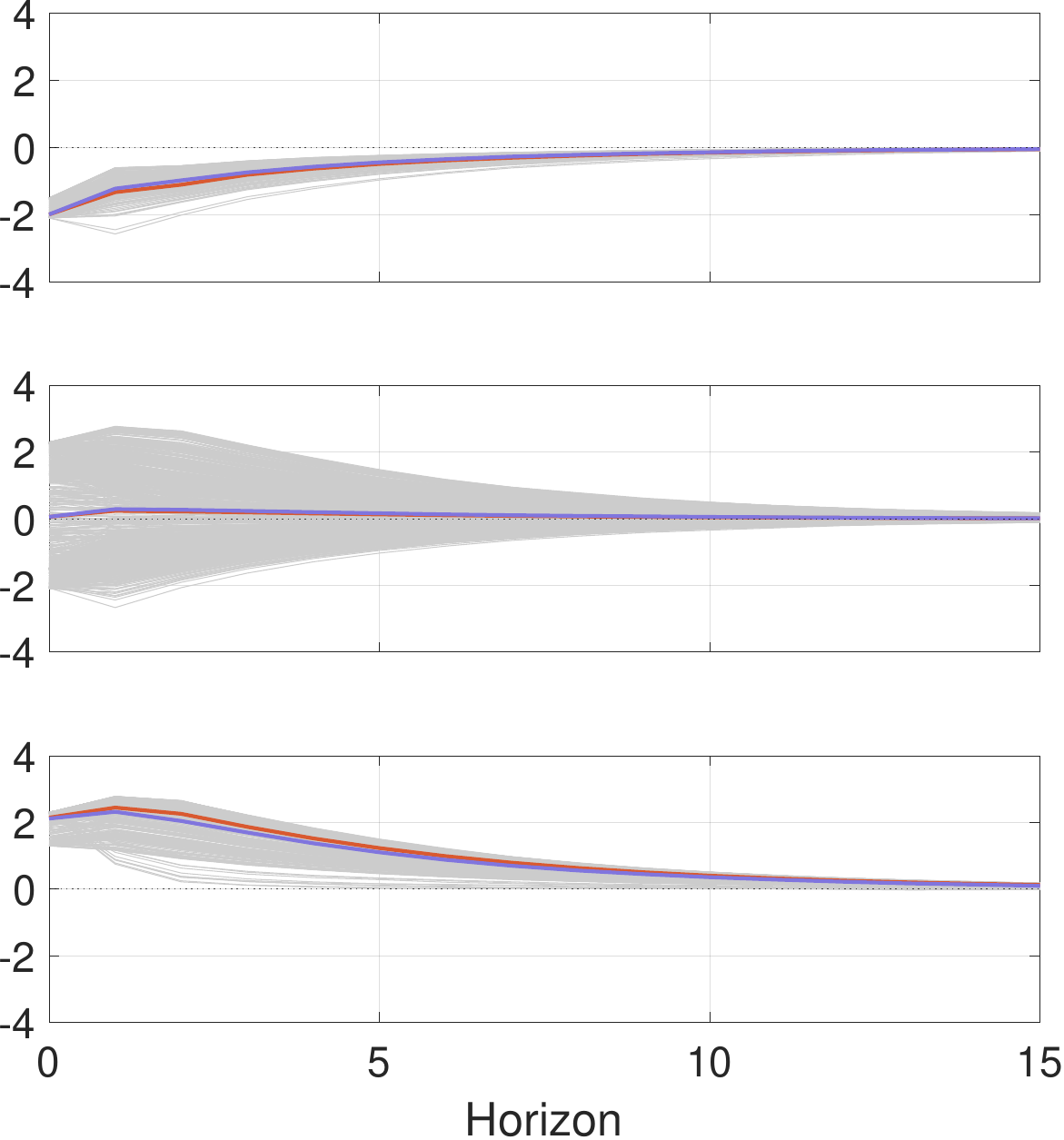}
        \label{fig:biv_smooth_3k}
    \end{subfigure}%
    \hfill%
    \begin{subfigure}[t]{0.32\textwidth}
        \centering
        \includegraphics[width=\textwidth]{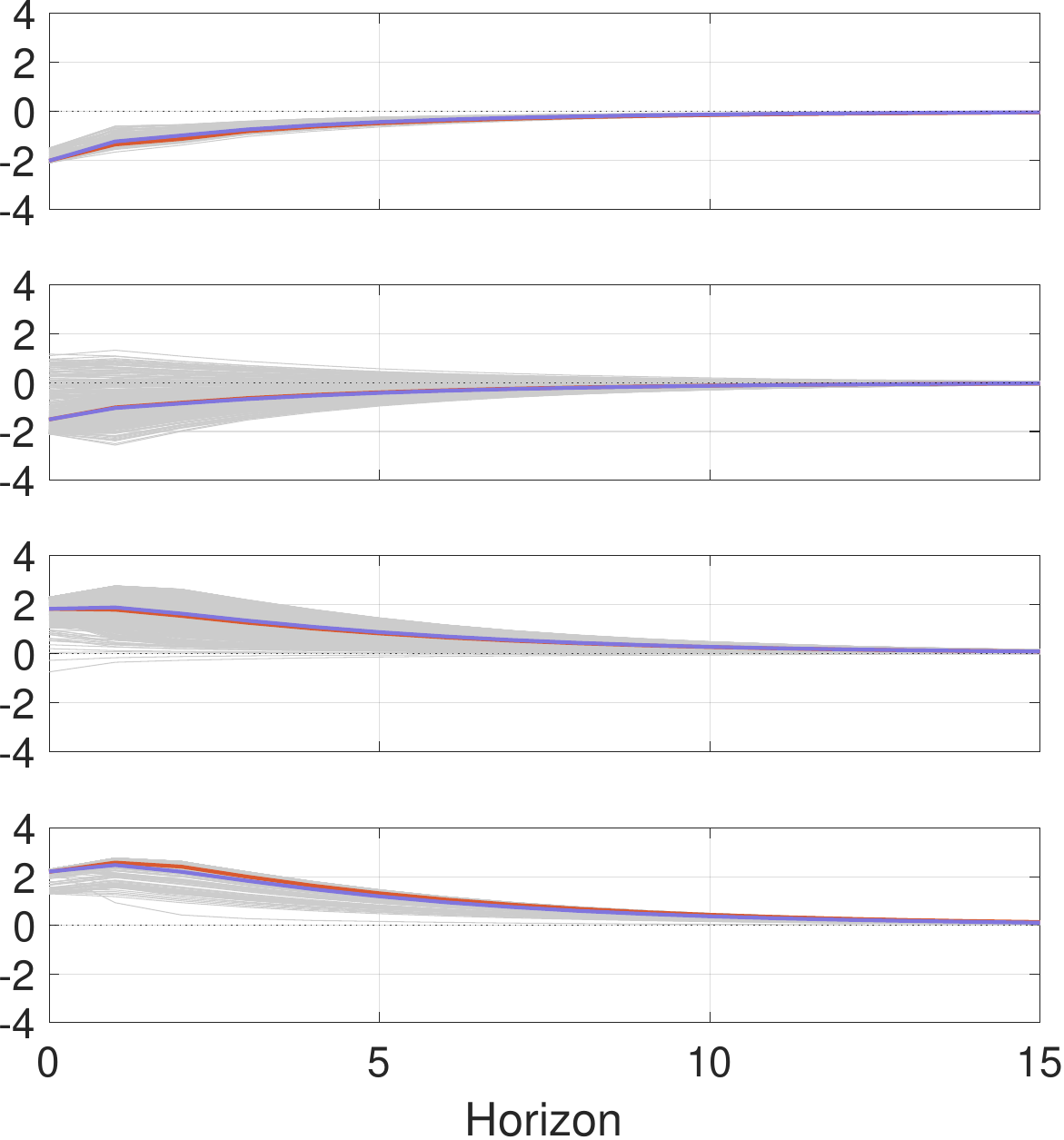}
        \label{fig:biv_smooth_4k}
    \end{subfigure}%
    \hfill%
    \begin{subfigure}[t]{0.32\textwidth}
        \centering
        \includegraphics[width=\textwidth]{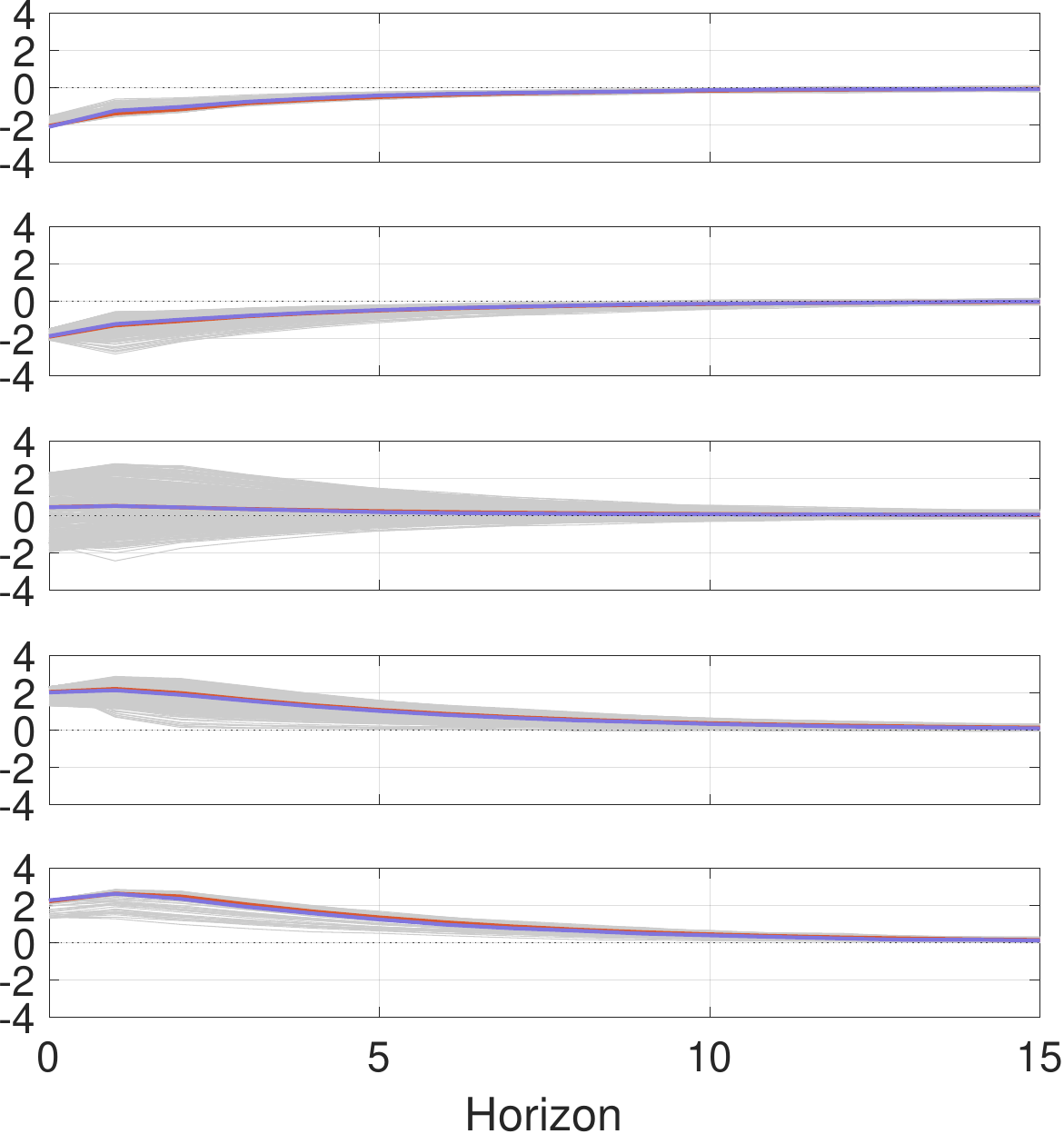}
        \label{fig:biv_smooth_5k}
    \end{subfigure}%
    \vspace{6pt}
    \begin{minipage}{\textwidth}
        \flushleft\footnotesize{\textit{Notes:} This figure illustrates the true impulse response functions for each $Z_{t-1}$ in gray, the clustered LP estimate in (purple), and the $CAR^h(\delta=1,k)$ (in red) for each partition ($K=3$ in the left panel, $K=4$ in the middle panel, and $K=5$ in the right panel).}
    \end{minipage}
\end{figure}

Figure \ref{fig:biSTM} illustrates how the iterative algorithm partitions the data into different subsamples, correspondingly, into different groups of IRFs for the bivariate DGP. The grouping of the IRFs (gray) appears to be closely linked to the impact responses. In addition, the overlap of the clustered LP and the $CAR^h(1,k)$ for each $k\in C_k$ evidences how the proposed estimator recovers the $CAR$. Lastly, in simulations with a smaller $\hat{K}$ (see left panel where $\hat{K}=3$), the estimator groups together very different dynamic IRFs, which results in clustered LP estimates that are close to zero in the short-run. In contrast, in cases where the number of groups is greater (see the right panel where $\hat{K}=5$), the estimation produces groups similar CARs.

\subsection{Alternative Nonlinear transformation of \texorpdfstring{$Z_{t-1}$}{Z(t-1)}}

Consider now the case where the autoregressive parameters evolve in a smooth and continuous manner, but in an asymmetric fashion. The DGP is given by:
\begin{align}
\beta_t(z_{t-1})   &= -0.4 + 0.7 |z_{t-1}| + \eta_{\beta,t}, \label{eq:absval_beta}\\
\gamma_{1,t}(z_{t-1}) &= -0.2 + 0.5 |z_{t-1}| + \eta_{g1,t}, \label{eq:absval_g1} \\
\gamma_{2,t}(z_{t-1}) &= -0.1 + 0.2 |z_{t-1}| + \eta_{g2,t}, \label{eq:absval_g2}
\end{align}
where $(\eta_{\beta,t},\;\eta_{g1,t},\;\eta_{g2,t})\sim \mathcal{N}(\mu\mathbf{1},\sigma_{\text{p}}^2\mathbf{I})$, are serially uncorrelated with  $\mu=0$, and $\sigma_{\text{p}}^2=0.0009$. As in the univariate smooth transition model, $\{z_t\}$ is assumed to follow an $\mathrm{ARMA}(p_z,m_z)$ with AR parameters
$(\phi_1^z, \phi_2^z) = (0.6, 0.3)$, MA parameters
$(\theta_1^z, \theta_2^z, \theta_3^z) = (0.8, 0.7, 0.2)$, $\mu_z=1$
and $\sigma_z = 0.003$. Although this DGP does not feature discrete regime switching, it serves to illustrate how alternative forms of time-variation can still be meaningfully grouped into a few clusters.  

The third column of Table~\ref{tab:k_counts} reports the selection frequencies for this DGP where the nonlinear transformation of the driving variable is $|z_{t-1}|$. Note that in this case, where time-variation evolves smoothly and there is no true number of clusters, the iterative procedure favors a small number of clusters: the classification frequency is 13.51\% for $K=2$, 73.14\% for $K=3$, and 12.8\% for $K=4$. The selection frequency for $K>4$ is negligible.

\begin{figure}[htbp]
    \centering
    \captionsetup{position=top}
    \caption{Absolute Value Model}
    \label{fig:abs}
    \begin{subfigure}[t]{0.32\textwidth}
        \centering
        \includegraphics[width=\textwidth, totalheight=7cm, keepaspectratio]{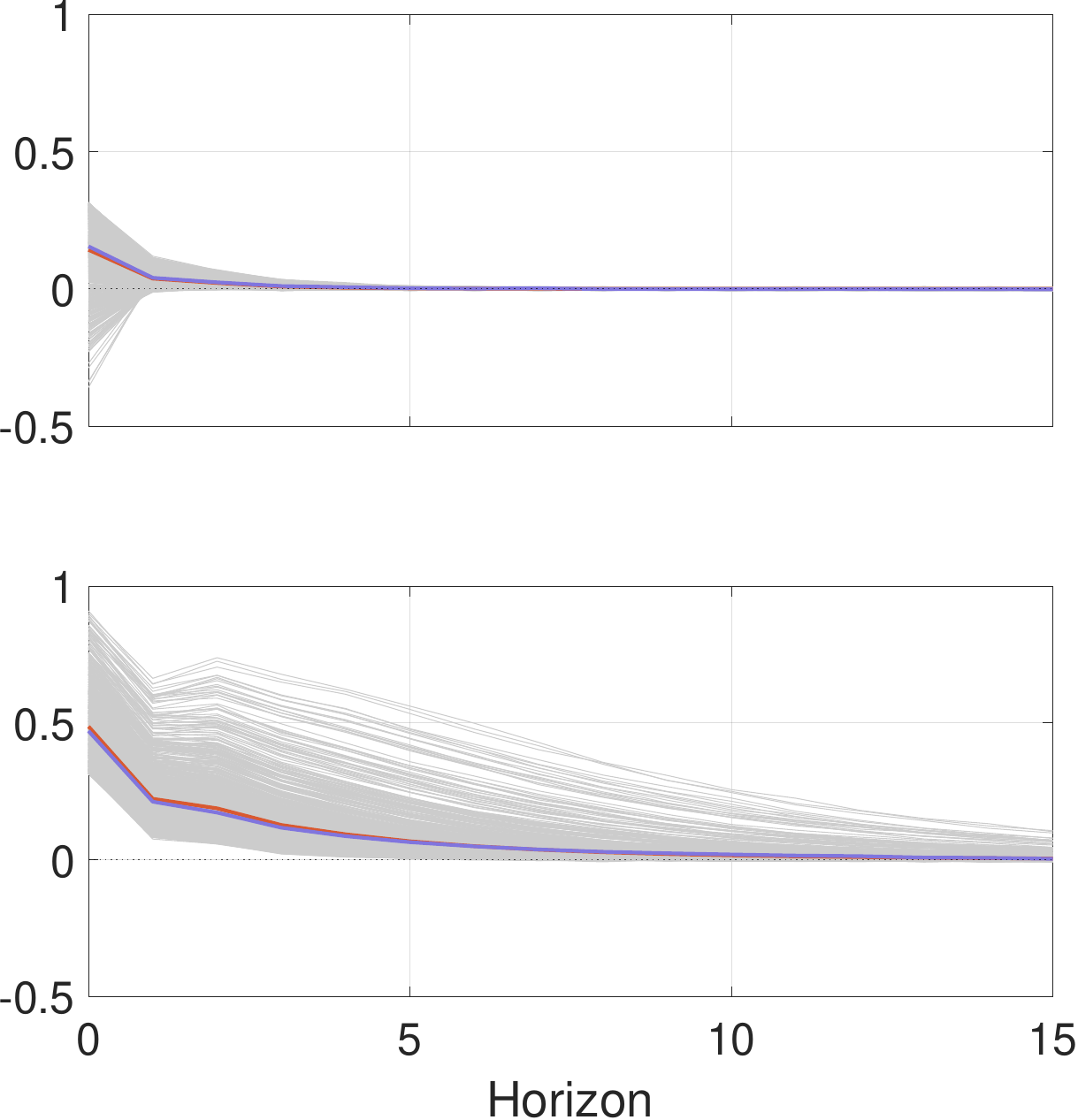}
        \label{fig:abs_2k}
    \end{subfigure}%
    \hfill%
    \begin{subfigure}[t]{0.32\textwidth}
        \centering
        \includegraphics[width=\textwidth, totalheight=7cm, keepaspectratio]{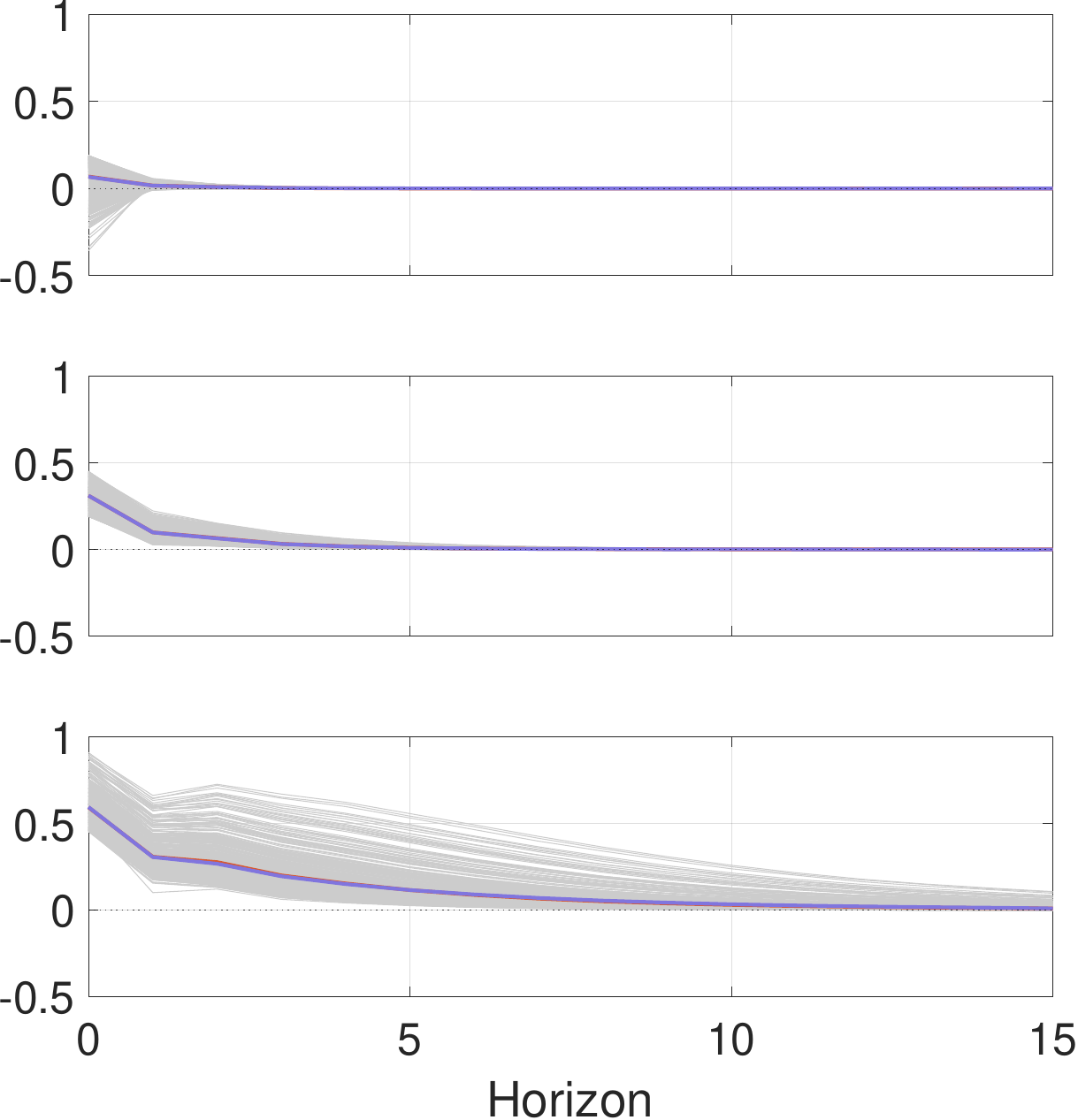}
        \label{fig:abs_3k}
    \end{subfigure}%
    \hfill%
    \begin{subfigure}[t]{0.32\textwidth}
        \centering
        \includegraphics[width=\textwidth, totalheight=7cm, keepaspectratio]{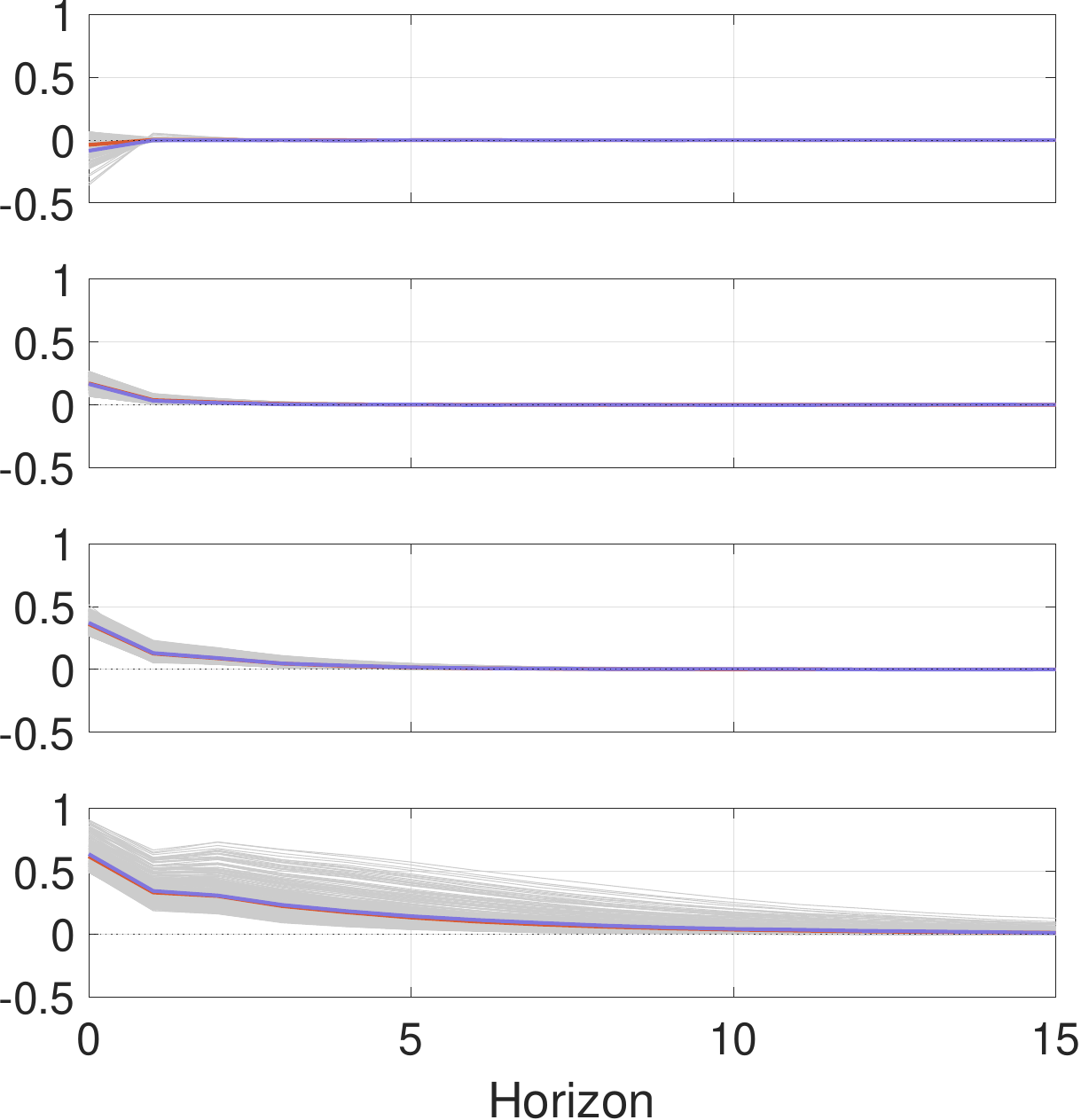}
        \label{fig:abs_4k}
    \end{subfigure}%
    \vspace{6pt}
    \begin{minipage}{\textwidth}
        \flushleft\footnotesize{\textit{Notes:} This figure illustrates the true impulse response functions for each $Z_{t-1}$ in gray, the clustered LP estimate in (purple), and the $CAR^h(\delta=1,k)$ (in red) for each partition ($K=2$ in the left panel, $K=3$ in the middle panel, and $K=4$ in the right panel).}
    \end{minipage}
\end{figure}

Figure~\ref{fig:abs} displays simulation results for this alternative DGP . The response corresponding to each value of $z_{t-1}$ is shown in gray, the $CAR^h(1,k)$ for each $k \in  C_k$ is colored red, and the clustered LP estimates are colored purple. The key takeaways from these simulations echo the insights derived from the smooth threshold models. The clustered LP estimates provide a good approximation of the conditional average responses. Due to the stationarity of the DGP, the effect of the shock to $\varepsilon^x_t$ dies quickly; therefore, differences between responses across groups are concentrated on impact and on the short term. In summary, the three DGPs illustrate how the clustered LP provides a causal summary of the time-varying effect of interest.

\subsection{Identifying Heterogeneity in Small Samples}\label{sec:small}

Our iterative algorithm partition has two components, the initial grouping of the data via k-means on $Z_{t-1}$ and the evaluation step that compares IRFs between groups, leading to the selection of responses that are statistically different given a horizon of interest $\tilde H$. Figure \ref{fig:histo} reports the frequency of $\hat{K}$ for each of the DGPs, setting $\tilde H={0,5}$ in the evaluation step, and samples of size $T={500,1000,200}$.

\begin{figure}[htbp]
    \centering
    \caption{Frequency of $\hat{K}$ across DGPs for different sample sizes}
    \includegraphics[width=0.9\linewidth]{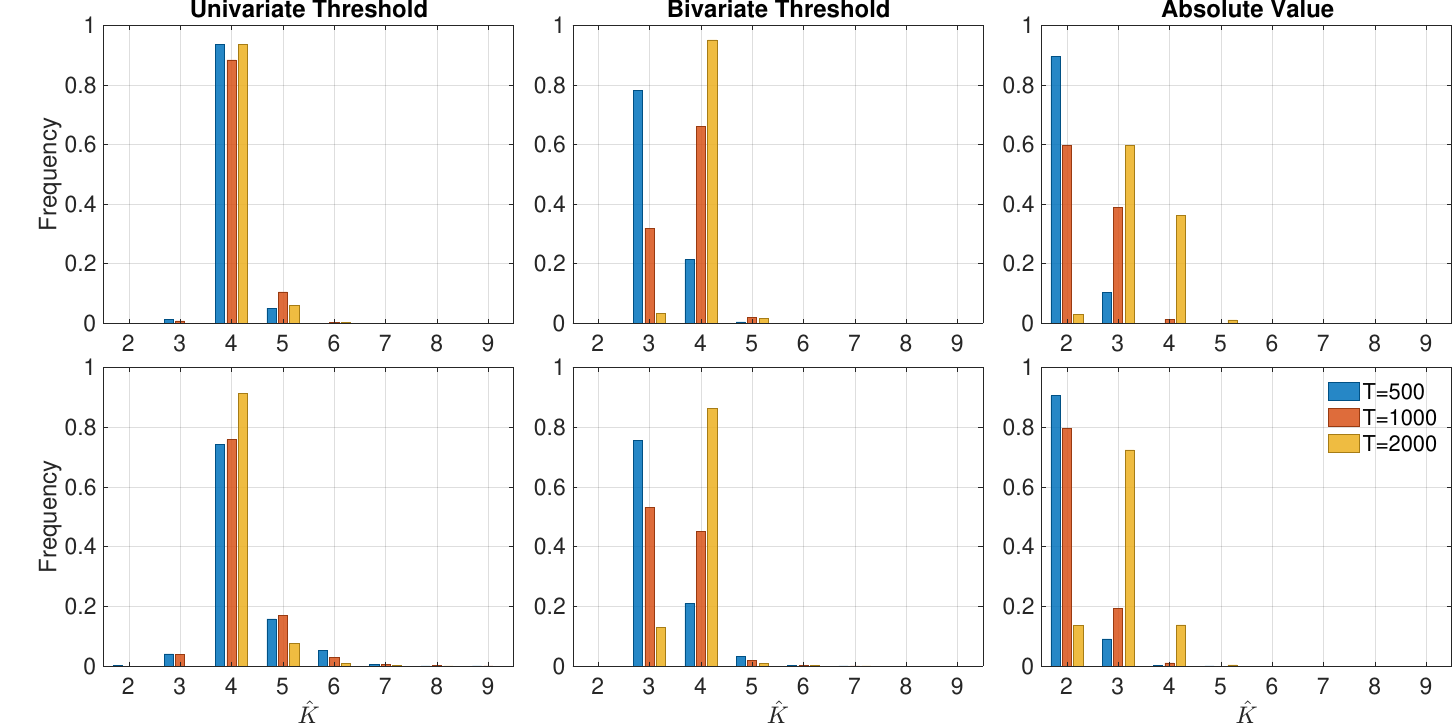}
     \label{fig:histo}
\begin{minipage}{\textwidth}
        \flushleft\footnotesize{\textit{Notes:} This figure illustrates the frequency of the estimated $\hat{K}$ across Monte Carlo simulations for different sample sizes and two values of $H$ in the evaluation step.}
    \end{minipage}    
\end{figure}

For the threshold DGPs, when $\tilde H=0$ the frequency of $\hat{K}=4$ is  higher than for $\tilde H=5$, especially for smaller sample sizes. In other words, if the evaluation step (i.e., the Wald test) is based only on the impact response, then the differences in the IRFs are mainly due to the distinct initial $Z_{t-1}$. In contrast, when the horizon used in the evaluation step increases ($\tilde H=5$), the system dynamics play a larger role in accounting for differences in the IRFs. For the multivariate  threshold and the absolute value DGPs, a key takeaway is that the smaller the number of observations, the more likely the procedure is to select a small number of clusters. This suggests that the iterative algorithm is conservative when the data is limited and avoids splitting the sample into clusters with few observations.

\section{The Uncertainty Channel of Monetary Policy Transmission to Treasury Yields}\label{sec:application}
A central question in macroeconomics is whether \textit{monetary policy uncertainty} (i.e., uncertainty around the central bank's course of action) affects the transmission of monetary policy to financial assets. Several empirical studies have found evidence of an uncertainty channel whereby heightened uncertainty dampens the response of U.S. Treasury yields to monetary policy shocks (see \citet{Tillman2020}, \citet{Bauer2021}, and \citet{DePooter2021}). Specifically, using an event study around FOMC dates, \citet{DePooter2021} show that the pass-through of monetary policy to medium- and long-term yields is stronger when monetary policy uncertainty is low, with the effect operating mainly through the term premium. They attribute this effect to investors taking larger positions when uncertainty is low; therefore, when a shock occurs, investors are forced to adjust more abruptly, pushing the term premium and yields up by a greater amount. Similar results are found by \citet{Bauer2021} who link the muted response to a signal extraction problem. That is, under high uncertainty, investors attribute less weight to signals from the Fed and, consequently, monetary policy is less effective. Using state-dependent local projections, \citet{Tillman2020} finds that monetary pollicy transmission to medium- and long-term yields is weaker under high uncertainty. Flight to safety, in the form of heightened investor appetite for the safety of locking in long-term bonds over repeatedly rolling over short-term debt, drives down the compensation required holding longer maturities. 

Less is known about how \emph{macroeconomic uncertainty} (i.e., uncertainty about the macroeconomic environment at a given point in time) affects the effectiveness of monetary policy. While increases in \textit{macroeconomic uncertainty} may be linked to increased \textit{monetary policy uncertainty}, other factors such as geopolitical tensions, fiscal policy uncertainty, and real economic shocks can lead to greater macroeconomic uncertainty. Hence, measures of monetary policy and macroeconomic uncertainty do not necessarily move in a synchronous manner. For instance, following the onset of the Global Financial Crisis, the Fed signaled its commitment to expansionary policy, and monetary policy uncertainty declined even as uncertainty about the macroeconomic outlook remained elevated. In contrast, monetary policy uncertainty remained high during the post-COVID inflation surge, while uncertainty about the macroeconomic environment dropped.

We employ our clustered LP method to investigate whether the existing literature on the uncertainty channel of monetary policy transmission to U.S. treasury yields overlooked potential interactions between these different types of uncertainty, which could give rise to time-varying responses. Our approach has several advantages relative to the event studies and state-dependent methods: it does not require the researcher to impose a-priori restrictions about the classification and number of the states; it groups the data taking into account the behavior of both macroeconomic and monetary policy uncertainty, and it does not take a stance on the functional form that drives the time variation but explicitly links this time variation to changes in uncertainty.

\subsection{Data Description and Preliminaries}

We use monthly data for the U.S. that span February 1988 to December 2023. The daily off-the-run, nominal Treasury zero-coupon bond yields (\citet{GURKAYNAK20072291}) are obtained from the Federal Reserve.\footnote{The data are available at https://www.federalreserve.gov/pubs/feds/2006/200628/200628abs.html.} The decomposition of daily yields into a term premium, a future short-term interest rate expectations component, and a residual follow \citet{christensen2011}; the data are obtained from the Federal Reserve Bank of San Francisco.\footnote{The data are available at: https://www.frbsf.org/research-and-insights/data-and-indicators/treasury-yield-premiums/} To construct monthly series, we aggregate the yield, term premium, and expectations component by averaging the monthly observations. 

To measure macroeconomic uncertainty, we use the monthly 3-period ahead macroeconomic uncertainty index by \citet{jurado2015}. This index captures the common variation in \emph{macroeconomic uncertainty} (defined as the conditional volatility of the purely unforecastable component of a series) across many macroeconomic variables, such as real output and income, employment and hours, real retail, manufacturing, and trade sales, and consumer spending. We employ \citet{hrs_mpu} \textit{monetary policy uncertainty} index, a monthly textual-based index based on articles published in the New York Times, Wall Street Journal, and Washington Post. We standardize both uncertainty indices to facilitate comparison. Data for the CPI inflation (constructed from monthly differences in the CPI series, in percent), and the unemployment rate are obtained from FRED.

The monetary policy shock is identified using a high-frequency approach based on \citet{bauer2023}, who use changes in federal funds futures prices in a narrow window around FOMC announcements. Even though we focus on the response of the 5- and 10-year Treasury yields, the Online Appendix C reports results for the 1-year and 2-year yields. We also present results for alternative identification approaches; namely the daily change of 2-year Treasury yields on FOMC days as in \citet{Tillman2020} and the \citet{bauer2023} FOMC frequency series, both aggregated to a monthly series following the weighting scheme in \citet{kilian2024}.

To gather some preliminary insights regarding the connection between \textit{macroeconomic uncertainty}, \textit{monetary policy uncertainty} --hereafter \textit{MPU} and \textit{MacroUncer}, respectively-- and the effect of monetary policy shocks on the 5-year yield, we estimate the state-dependent LP:
\begin{equation*}
y_{t+h} = I_{t-1}\left[\alpha^h_{(1)}+ \beta^h_{(1)}\varepsilon_t+ \boldsymbol{\pi}^{h}_{(1)}\mathbf{w}_t\right] + (1 - I_{t-1})\left[\alpha^h_{(0)}+ \beta^h_{(0)}\varepsilon_t + \boldsymbol{\pi}^{h}_{(0)} \mathbf{w}_t\right]
+ u_{t+h},
\label{eq:sd_lp}
\end{equation*}

where $y_{t+h}$ denotes the $5$-year Treasury yield at horizon $h$, $\varepsilon_t$ is the monetary policy shock, $\mathbf{w}_t$ is a vector of control variables, $I_{t-1}$ is a binary state indicator that takes on the value of 1 when $z_{t-1}> 0$ ($z_{t-1}$ equal $MPU_{t-1}$ or $MacroUncer_{t-1}$) and zero otherwise. $\mathbf{w}_t$ includes the first lags of CPI inflation, unemployment, and the 5-year yield. To investigate possible time-variation in the responses, we estimate the model using four expanding subsamples (1988:2-2015:12, 1988:2-2019:12, 1988:2-2021:12, and 1988:2-2023:12) with the last period corresponding to the full sample. For comparison, we report the estimate IRF for a linear LP specification.

\begin{figure}[h!]
  \centering
  \captionsetup{position=top}
  \caption{State-Dependent LP: Monetary Policy Uncertainty}
  \includegraphics[width=0.75\textwidth, trim=0cm 0cm 0cm 0cm, clip]{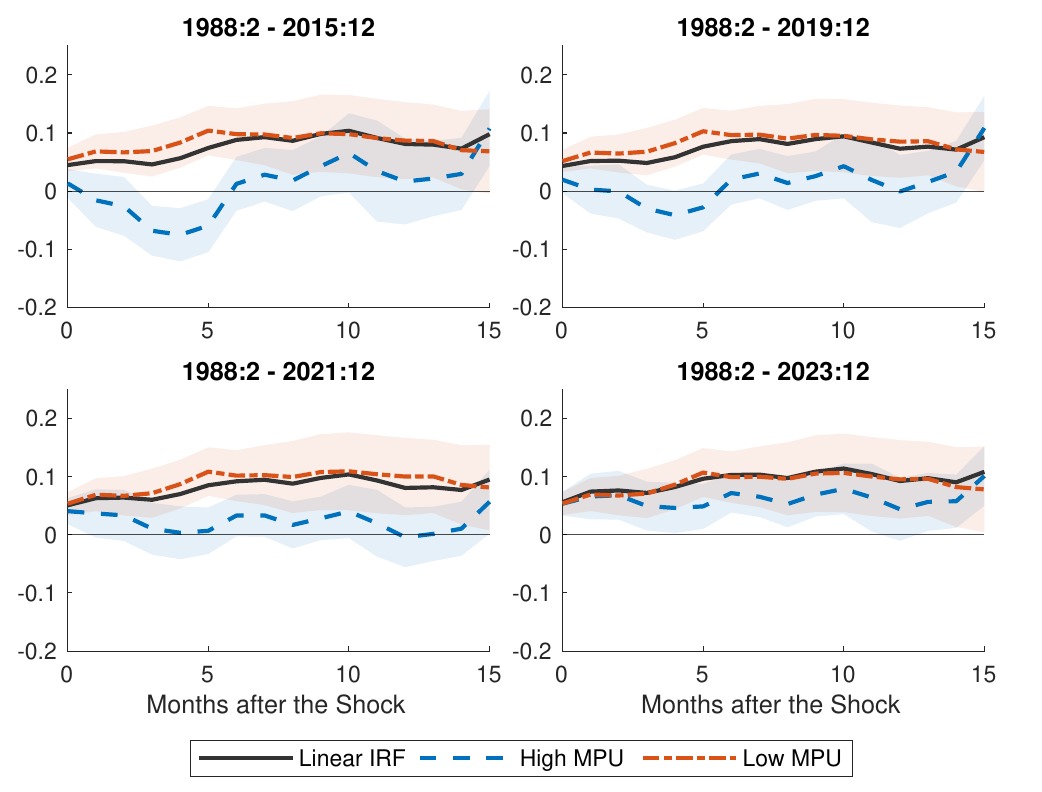}
  \vspace{8pt}
  \begin{minipage}{0.75\textwidth}
    \flushleft\footnotesize{\textit{Notes:} State-dependent Impulse Response Functions of the 5-year yield to a one-standard deviation monetary policy shock under different monetary policy uncertainty regimes, above (blue line) and below (red line) the mean, with 68\% confidence intervals. The impulse response function of the linear model (in black) is plotted for comparison. Each panel represents a different sample.}
  \end{minipage}
  \label{fig:state_dep_MPU}
\end{figure}

\begin{figure}[h!]
  \centering
  \captionsetup{position=top}
  \caption{State-Dependent LP: Macroeconomic Uncertainty}
  \includegraphics[width=0.75\textwidth, trim=0cm 0cm 0cm 0cm, clip]{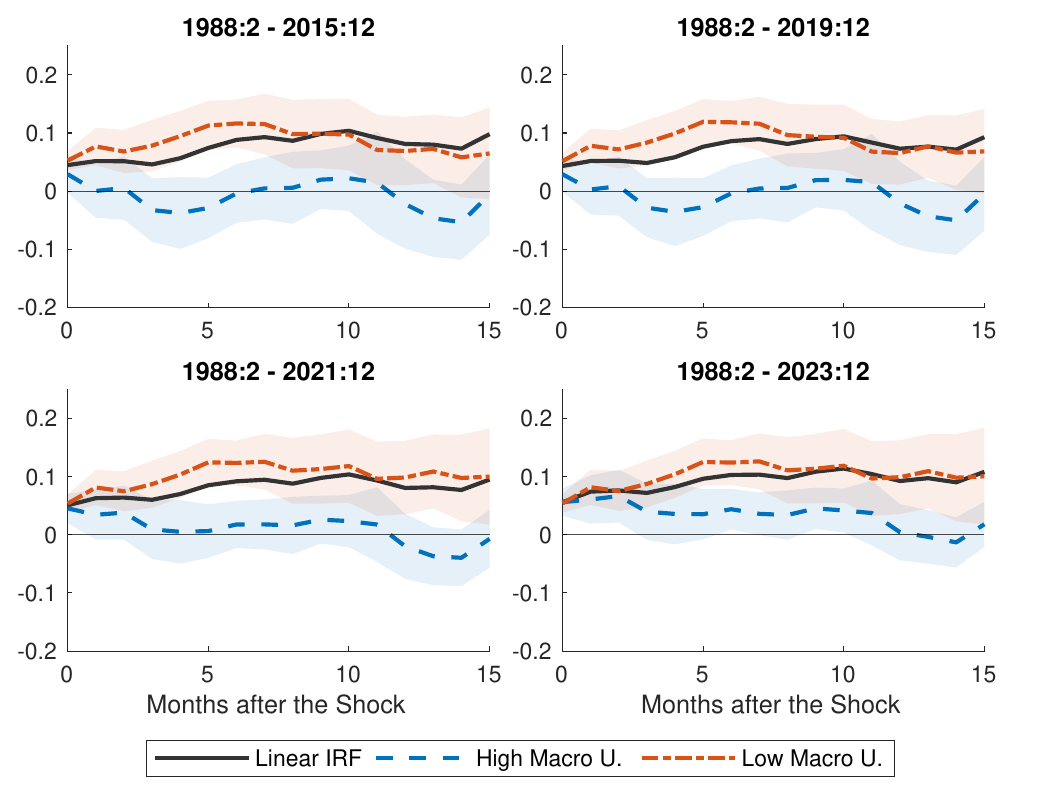}
  \vspace{8pt}
  \begin{minipage}{0.75\textwidth}
    \flushleft\footnotesize{\textit{Notes:} State-dependent Impulse Response Functions of the 5-year yield to a one-standard deviation monetary policy shock under different macroeconomic uncertainty regimes, above (blue line) and below (red line) the mean, with 68\% confidence intervals. The impulse response function of the linear model (in black) is plotted for comparison. Each panel represents a different sample.}
  \end{minipage}
  \label{fig:state_dep_MU}
\end{figure}

Figure \ref{fig:state_dep_MPU} shows the response to a one standard deviation monetary policy shock (5.5 basis points) during the high (dashed blue line) and low (dot-dashed red line) \textit{MPU} states, as well as the response estimated with the linear LP (black solid line). Shaded areas represent 68\% confidence intervals.\footnote{Confidence intervals are computed using HAC standard errors.} The difference between the response in high- and low-uncertainty states appears to be time-varying. Whereas estimates for the earlier samples reveal a clear difference between the IRFs in both states (negative during high MPU and positive during low MPU), this difference is muted in the full sample. This suggests a change in the interaction between MPU and the transmission of monetary policy since (or during) the Covid-19 pandemic. Moreover, the linear LP overestimates the effectiveness of monetary policy during high MPU times.

The estimates reported in Figure \ref{fig:state_dep_MU} are indicative of time-variation in the interaction between \textit{MacroUncer} and the transmission of monetary policy shocks to the 5-year yield. Note how the gap between the IRFs for low- and high-uncertainty states decreases as we extend the sample. For the majority of the subsamples, the response in times of high \textit{MacroUncer} is statistically insignificant or small at most horizons, while positive when \textit{MacroUncer} is low. Once the Covid-19 period is included (1988:2-2021:12 and 1988:2-2023:12), the response on impact and at short horizons turns positive. \\ To summarize, state-dependent LP estimates suggest that the uncertainty channel of monetary policy transmission is time-varying and multifaceted, with high monetary policy and macroeconomic uncertainty attenuating the effectiveness of monetary policy. 
\subsection{How effective is monetary policy in uncertain times?}

We investigate the role of both forms of uncertainty of the yields response to monetary policy via the clustered local projection model described in (\ref{eq:CLP}), where $Z_{t-1}$ includes the \textit{MPU} and the \textit{MacroUncer} indices, and the controls comprise lagged CPI inflation, lagged unemployment and the first lag of the dependent variable. The iterative procedure starts with 10 clusters and, testing the null of pairwise equality for $h=0,1, ...,15$ in the third step of the algorithm, yields a total of $\hat{K}=4$ clusters.

\begin{table}[h!]
  \centering
  \caption{Estimated cluster means and their labeling}
  \label{tab:clusters}
  \vspace{0.25em}
  \begin{tabular}{lcc}
    \toprule
    \textbf{Cluster} & \textbf{Mean Macro U.} & \textbf{Mean MPU} \\
    \midrule 
    \color{MatlabBlue} 1: Low macro uncertainty, low MPU     & -0.27 & -0.67 \\
    \color{MatlabPurple} 2: Low macro uncertainty, moderate MPU& -0.23 & 0.32 \\
    \color{MatlabYellow} 3: Moderate macro uncertainty, \;high MPU & 0.63 & 2.29 \\
    \color{MatlabRed} 4: High macro uncertainty, low MPU & 2.73 & -0.34 \\
    \bottomrule
  \end{tabular}
\end{table}

Table \ref{tab:clusters} summarizes the means for each of the four clusters and their labeling. Recall that we standardize both indices, thus positive values represent uncertainty above the historical mean (hereafter high uncertainty) and negative values represent uncertainty below the mean (hereafter low uncertainty). Note how the different groups have distinctive characteristics. For instance, cluster one corresponds to macro and monetary policy uncertainty below the mean whereas cluster three corresponds to moderate macro and high monetary policy uncertainty. 

\begin{figure}[h!]
  \caption{Classification of Macroeconomic and Monetary Policy Uncertainty Indices}
  \label{fig:unc_scatters}
  \centering
  \vspace{-8pt}
  \includegraphics[width=0.48\linewidth]{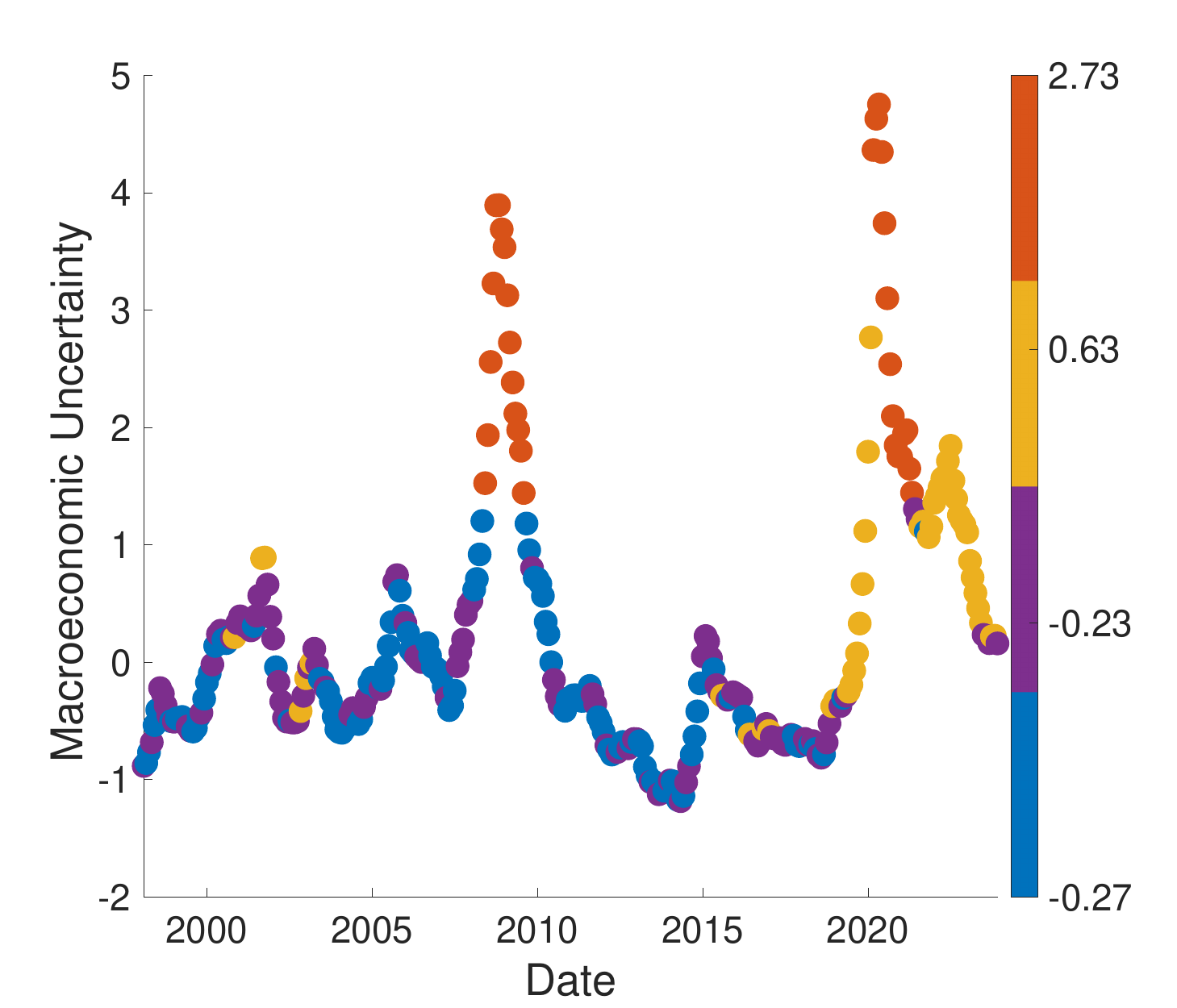}
  \includegraphics[width=0.48\linewidth]{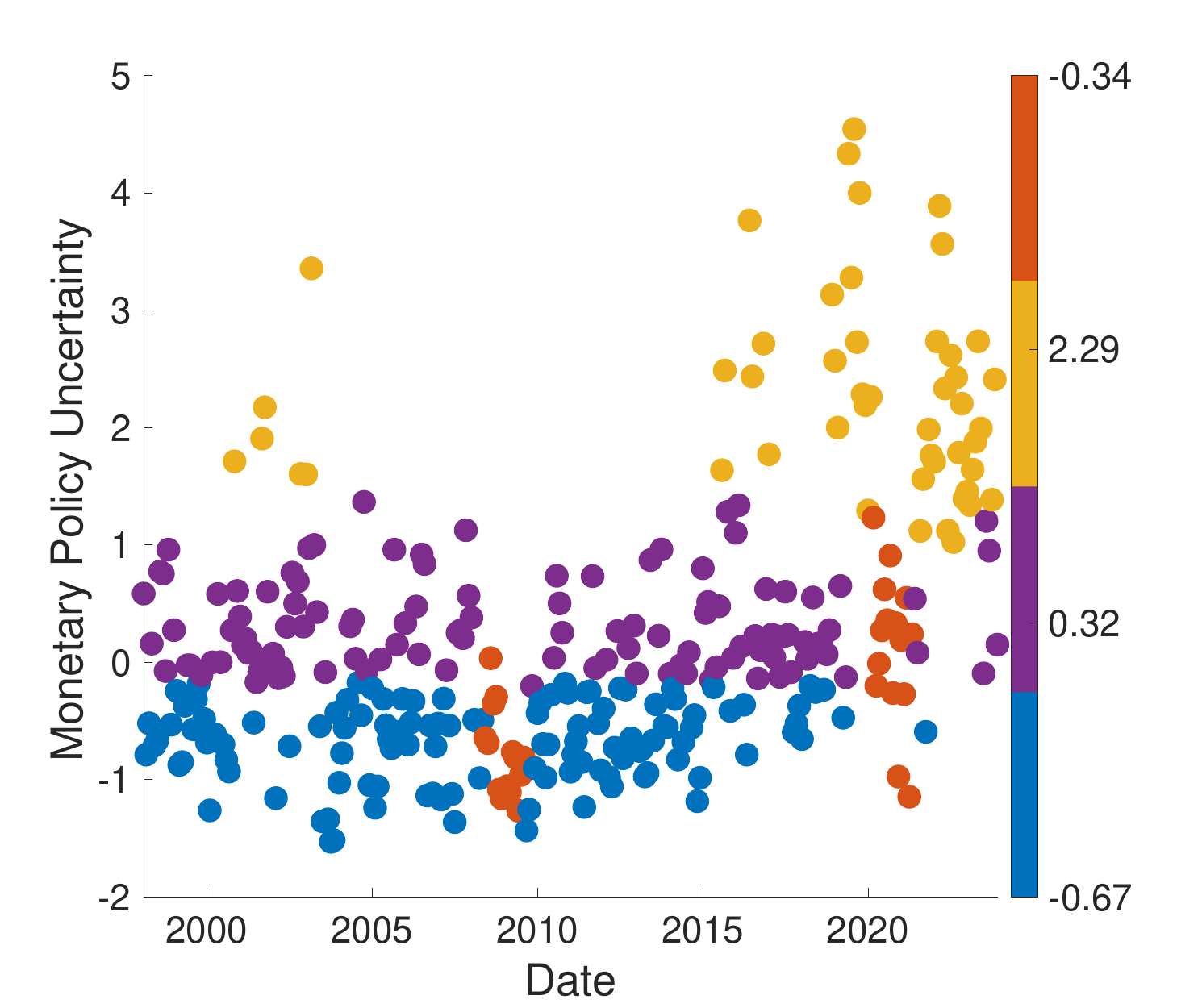}
  \vspace{-8pt}
    \begin{minipage}{\textwidth}
        \flushleft\footnotesize{\textit{Notes:} This figure reports the classification of the macroeconomic (left panel) and monetary policy uncertainty (right panel) indices into the $K=4$ clusters.}
    \end{minipage}
\end{figure}

To visualize how the historical data are partitioned by the iterative procedure, the scatterplots in Figure \ref{fig:unc_scatters} map the uncertainty data series to the four clusters. As the figure illustrates, some periods of high \textit{MacroUncer} (e.g., the Great Recession and the Covid-19 pandemic) were also periods of low \textit{MPU} (cluster 4 in red). In contrast, the months preceding and following the pandemic correspond to high \textit{MPU} and moderate \textit{MacroUncer} (cluster 3 in yellow). This cluster includes the period when banks took initial write-downs (October 2007) and the Fed's announcement of QE2 (November 2010). Furthermore, note that clustering based on only one of the two uncertainty indices--as effectively done in the state-dependent models--would overlook periods where monetary policy and macroeconomic uncertainty do not move in a synchronous manner.

We now turn to the cluster LP estimates. Figure \ref{fig:appl_clustered_LP} depicts the IRFs of the 5- and 10-year U.S. Treasury nominal yields and their components to a one–standard deviation contractionary monetary policy shock for each of the four clusters. For comparison, we also plot the response estimated by a linear LP (solid black line). Our results reveal an interesting pattern regarding the relative importance of each type of uncertainty for the transmission of monetary policy shocks. On impact and in the very short-run, \textit{MacroUncer} is the main driver of heterogeneity in the response of the yields. In fact, the impact responses during high \textit{MacroUncer} (clusters 3 and 4 in solid yellow line, and dash-dotted red line, respectively) are positive and statistically significant, whereas the responses in times of low \textit{MacroUncer} (clusters 1 and 2 in dotted blue line, and dashed purple line, respectively) are not. However, at longer horizons, \textit{MPU} drives the responsiveness of yields. Note how the response of the yield is an order of magnitude higher twelve months after the shock when \textit{MPU} is low (clusters 1 and 4) than when it is high (clusters 2 and 3). Of particular interest is the linear LP estimate, which corresponds to the average IRF over the sample and underestimates (overestimates) the response of the yields during low \textit{MPU} and overestimates the response during low \textit{MacroUncer} and moderate \textit{MPU}.

\begin{figure}[h!]
  \centering
  \caption{The Effect of Contractionary Monetary Policy Shocks Across Uncertainty Regimes}
  \includegraphics[width=0.75\textwidth, trim=0cm 0cm 0cm 0cm, clip]{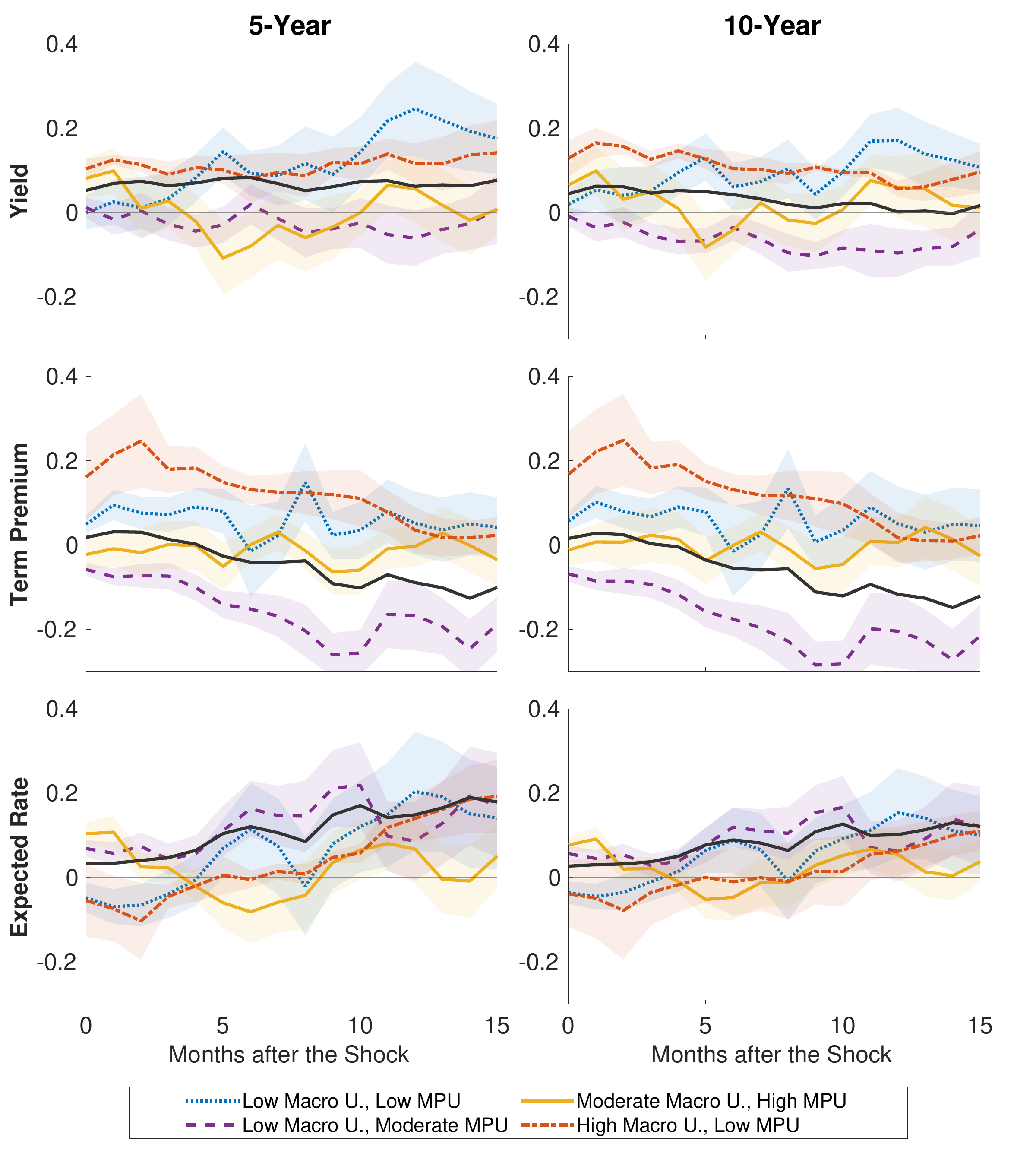}
  \vspace{-8pt}
  \begin{minipage}{0.75\textwidth}
    \flushleft\footnotesize{\textit{Notes:} The figure reports clustered Local Projections Impulse Response Functions of the 5- and 10-year yield and its components to a one-standard deviation monetary policy shock under different monetary policy and macroeconomic uncertainty regimes. Shaded areas represent 68\% confidence intervals. The estimated linear LP response is plotted in black.}
  \end{minipage}
  \label{fig:appl_clustered_LP}
\end{figure}

To dig deeper into the mechanisms that drive this heterogeneity, we estimate the responses of the term premium and the expected rate component. On impact, the effect of a contractionary monetary policy shock on the term premiums is heterogeneous across clusters, but the average responses of the expectations component are similar for clusters with high (low) \textit{MPU}. To gather some insight regarding the role of heightened uncertainty in the transmission of monetary policy, let us focus on the states where either of the indices is high. Comparing the average responses of the yields and the expectations components during times of high \textit{MPU} and moderate \textit{MacroUncer} (solid yellow line) reveals that short-term dynamics are driven by the expectations component. In contrast, the positive response of the term premium explains the short-run increase experienced by the yields during times of high \textit{MacroUncer} (dash-dotted red line); indeed, the increase in the term premium more than offsets the decline in the expected rate on impact and on the first three months following the shock. Moreover, the linear LP underestimates (overestimates) the response of the term premium (expectations during times of low \textit{MPU}.

Our estimates suggest that the mechanisms through which macroeconomic and monetary policy uncertainty affects each component of the yield differ. The behavior of the term premium is consistent with the view that macroeconomic uncertainty amplifies the risk compensation required by investors following a contractionary shock. In times of high macroeconomic uncertainty, the distribution of possible future states is more widespread: growth, inflation, and economic policy could swing in multiple directions, making the future path of short-term rates inherently harder to predict. Even a surprise rate hike may increase market disagreement---some investors may fear an overshoot into recession, while others may interpret it as a signal of stronger inflation risks---further increasing perceived risk around the future rate path and pushing the term premium upward. This explains why clusters 3 (4), characterized by moderate (high) \textit{MacroUncer}, display the largest term premium responses. Nevertheless, in times of similar macroeconomic uncertainty, \textit{MPU} is associated with a dampening force on the term premium, consistent with the flight-to-safety channel highlighted by \citet{Tillman2020}: when the future course of policy is unclear, long-term securities become relatively more attractive as safe havens, compressing the additional compensation demanded by investors. This is reflected in the smaller term premium response of cluster 3 relative to cluster 4, and of cluster 2 relative to cluster 1.

Taken together, these results suggest that macroeconomic and monetary policy uncertainty operate through complementary but distinct channels: the former primarily amplifies the risk compensation embedded in the term premium, while the latter governs the speed and persistence with which markets revise their expectations about the future rate path following a monetary policy shock. Overall, this analysis demonstrates the potential of our IRF estimation approach to capture the nonlinearity in the interaction between macroeconomic and monetary policy uncertainty, offering insights that would be difficult to obtain from standard two-state models conditioning on a single uncertainty measure or a linear LP. 

\section{Conclusion}\label{sec:conclusions}
This paper proposed a new method to estimate impulse response functions in time-varying parameter models, which we term clustered local projections (LP). We showed that the clustered LP recover the conditional average response $CAR$ if the set of variables that drive the time-variation are exogenous and the conditional marginal response $CMR$ when they are endogenous. Thus, the impulse response estimands from the clustered LP provide a causal summary of the nonlinear effect in each cluster. 

From a methodological perspective, we proposed an iterative estimation approach that comprises three steps: classification using the k-means algorithm, joint estimation of impulse responses for all clusters and horizons via local projections, and evaluation. Monte Carlo simulations demonstrated that the clustered LP estimator recovers accurately $CAR$ in smooth threshold models, as well as in specifications in which the parameters depend on the driving variable in an asymmetric fashion. They also illustrated the performance of the iterative algorithm in smaller samples and when alternative horizons are employed in the evaluation step.

We used clustered LP to study how effective contractionary monetary policy is in shifting the 5- and 10-year US treasury nominal yields during uncertain times. Our estimates provided evidence that the mechanisms whereby macroeconomic and monetary policy uncertainty affect the U.S. Treasury yields differ. On the one hand, during times of high macroeconomic uncertainty, contractionary monetary policy is associated with a significant increase in the term premium, which, in turn, accounts for the increase in the yields. On the other hand, during times of high monetary policy uncertainty, contractionary policy leads to an increase in expected rates on impact and shortly after the shock, but to a decline at longer horizons. This response, in conjunction with a moderate increase in the term premium, accounts for the short-lived response of the yields.

\setstretch{1.15}

\bibliography{literature}
\bibliographystyle{jpe}

\newpage
\appendix
\section{Appendix}\label{sec:appendix}
\textbf{Proof of Proposition \ref{prop:CLP} }\label{app:proof}

\textbf{Part (i):} Because the clustered LP \eqref{eq:CLP} is fully interacted and by the FWL theorem together with Assumption~\ref{as:1}, the OLS estimand
$\beta_k^h$ equals the within-cluster regression coefficient:
\begin{equation}\label{eq:estimand_raw}
  \beta_k^h \;=\;
  \frac{\E[y_{t+h}\,\varepsilon_{t}\mid D_k=1]}
       {\E[\varepsilon_{t}^2\mid D_k=1]}.
\end{equation}
This is the population OLS coefficient from regressing $y_{t+h}$ on
$\varepsilon_t$ using only observations in the subsample
$\{t : D_k = 1\}$. In addition, since $D_k = \mathbf{1}\{Z_{t-1}\in\Ck\}$ is a deterministic function of
$Z_{t-1}$ and $\varepsilon_t$ is independent of $\mathcal{F}_{t-1}$
(Assumption~\ref{as:1}(a)), we have
$\varepsilon_t\indep Z_{t-1}$ and therefore
$\varepsilon_t\indep D_k$ for all $k$. The assumption \ref{as:1} and the timing of the clustering also ensure that a conditional independence assumption is satisfied so that  $ \varepsilon_t \;\indep\; U_{h,t+h} \;\mid\; D_k=1$. 

Using the structural representation \eqref{eq:structural} and the
conditional independence:
\begin{align}
  g_{h,k}(e)
  &\;\equiv\; \E[y_{t+h}\mid\varepsilon_t=e,\, D_k=1] \notag\\
  &\;=\; \E[\psi_h(e,U_{h,t+h})\mid\varepsilon_t=e,\, D_k=1] \notag\\
  &\;=\; \E[\psi_h(e,U_{h,t+h})\mid D_k=1]
  \;=\; \Psi_k^h(e). \label{eq:identification}
\end{align}
The second equality substitutes \eqref{eq:structural}. The third equality
uses $\varepsilon_t\indep U_{h,t+h}\mid D_k=1$: since $U_{h,t+h}$ is
independent of $\varepsilon_t$ given $D_k=1$, conditioning on
$\varepsilon_t=e$ does not alter the distribution of $U_{h,t+h}$. Thus,
the within-cluster regression function $g_{h,k}$ identifies the
conditional average structural function $\Psi_k^h$. The proof of Proposition \ref{prop:CLP} follows directly from Proposition 1 of \citet{kolesar2025dynamiccausaleffectsnonlinear} given our Assumption \ref{as:1}. In particular,
\begin{equation}\label{eq:KP_applied}
  \beta_k^h
  \;=\; \frac{\E[y_{t+h}\,\varepsilon_t\mid D_k=1]}
             {\E[\varepsilon_t^2\mid D_k=1]}
  \;=\; \int_I \omega_k(e)\, g_{h,k}'(e)\,de,
\end{equation}
where
\begin{equation}\label{eq:weight_proof}
  \omega_k(e) \;=\;
  \frac{\Cov\!\left(\mathbf{1}\{\varepsilon_t\geq e\},\,\varepsilon_t
        \,\Big|\, D_k=1\right)}
       {\Var(\varepsilon_t\mid D_k=1)}.
\end{equation}
By \citet{kolesar2025dynamiccausaleffectsnonlinear}, Proposition~1, $\omega_k(e)\geq 0$ for all $e$,
$\int_I\omega_k(e)\,de=1$, and $\omega_k$ is hump-shaped, peaking near
$\E[\varepsilon_t\mid D_k=1]$. Furthermore, $\omega_k$ depends only on
the distribution of $\varepsilon_t\mid D_k=1$, not on $y_{t+h}$ or~$h$. 

Substituting the identification result \eqref{eq:identification} into
\eqref{eq:KP_applied}:
\begin{equation}
  \beta_k^h
  \;=\; \int_I \omega_k(e)\,\Psi_k^{h\prime}(e)\,de,
\end{equation}
establishing \eqref{eq:part1}.

\textbf{Part (ii):} Part~(ii) follows from Part~(i) by showing that the conditional average
structural function $\Psi_k^h(e)$ is linear in~$e$ when $Z_t$ is exogenous. When $Z_t\indep\varepsilon_t$ for all $t$, the future path
$\{Z_t, Z_{t+1}, \ldots, Z_{t+h-1}\}$ does not depend on the realization
$\varepsilon_t = e$. Hence all time-varying coefficients
$\{\beta_t^h(Z_{t-1}), \gamma_t^h(Z_{t-1})\}_{h=0}^H$ are
independent of $e$. This is the direct analogue of \citet{Goncalves2024a}, where linearity of the potential outcome in $e$
follows from the absence of $e$-dependence in the future states. It
follows that $\psi_h(e, U_{h,t+h})$ is linear in $e$, e.g. $\psi_h\left(e, U_{h, t+h}\right)=A_h\left(U_{h, t+h}\right)+\Lambda_h\left(U_{h, t+h}\right) e$., where $A_h\left(U_{h, t+h}\right)$ and $\Lambda_h\left(U_{h, t+h}\right)$ are a random variables that depend on $U_{h,t+h}$ but
not on $e$ or $\delta$. Specifically, $\Lambda_h\left(U_{h, t+h}\right)$ aggregates the
cumulative effect of the shock through the chain of time-varying
coefficients along the path from $t$ to $t+h$. Then 
\begin{equation}\label{eq:linearity}
  y_{t+h}(e+\delta) - y_{t+h}(e) \;=\; \Lambda_h\left(U_{h, t+h}\right)\,\delta,
\end{equation}
In addition, by linearity, the conditional average structural
function is
\begin{equation}
  \Psi_k^h(e) \;=\; \E[\psi_h(e,U_{h,t+h})\mid D_k=1]
               \;=\; \E[\Lambda_h\left(U_{h, t+h}\right)\mid D_k=1]\cdot e,
\end{equation}
so its derivative $\Psi_k^{h\prime}(e) = \E[\Lambda_h\left(U_{h, t+h}\right)\mid D_k=1]$ is
\emph{constant} in~$e$. Substituting into \eqref{eq:part1} from
Part~(i):
\begin{equation}
  \beta_k^h
  \;=\; \int_I \omega_k(e)\,\Psi_k^{h\prime}(e)\,de
  \;=\; \E[\Lambda_h\left(U_{h, t+h}\right)\mid D_k=1]\cdot
        \underbrace{\int_I \omega_k(e)\,de}_{=\,1}
  \;=\; \E[\Lambda_h\left(U_{h, t+h}\right)\mid Z_{t-1}\in\Ck].
\end{equation}
Taking expectations in \eqref{eq:linearity} conditional on
$Z_{t-1}\in\Ck$:
\begin{equation}
  \CAR^h(\delta,k)
  \;=\; \E[\Lambda_h\left(U_{h, t+h}\right)\mid Z_{t-1}\in\Ck]\cdot\delta.
\end{equation}
Therefore
\begin{equation}
  \beta_k^h
  \;=\; \E[\Lambda_h\left(U_{h, t+h}\right)\mid Z_{t-1}\in\Ck]
  \;=\; \frac{\CAR^h(\delta,k)}{\delta}
  \;=\; \CAR^h(1,k)
  \;=\; \CMR^h(k)
\end{equation}
establishing \eqref{eq:part2}.

\end{document}